\documentclass[twoside,leqno]{article}

\usepackage[letterpaper]{geometry}

\usepackage{siam}

\usepackage[T1]{fontenc}
\usepackage{amsfonts}
\usepackage{graphicx}
\usepackage{epstopdf}
\usepackage{enumitem}
\usepackage{float}

\ifpdf
  \DeclareGraphicsExtensions{.eps,.pdf,.png,.jpg}
\else
  \DeclareGraphicsExtensions{.eps}
\fi

\newcommand{\cfold}{\text{cfold}}
\newcommand{\pfold}{\text{pfold}}

\newtheorem{conjecture}{Conjecture}
\newtheorem{fact}{Fact}
\newsiamremark{remark}{Remark}
\newsiamremark{hypothesis}{Hypothesis}
\crefname{hypothesis}{Hypothesis}{Hypotheses}
\newsiamthm{claim}{Claim}

\clearpage
\usepackage{amsopn}

\begin{document}

\title{\Large An Optimal Density Bound for Discretized Point Patrolling}

\author{Ahan Mishra\thanks{Department of Computer Science, Cornell University, Ithaca, NY 
(\email{abm247@cornell.edu}).}}
  
\date{}

\maketitle

\begin{abstract} 

The pinwheel problem is a real-time scheduling problem that asks, given $n$ tasks with periods $a_i \in \mathbb{N}$, whether it is possible to infinitely schedule the tasks, one per time unit, such that every task $i$ is scheduled in every interval of $a_i$ units. We study a corresponding version of this packing problem in the covering setting, stylized as the discretized point patrolling problem in the literature. Specifically, given $n$ tasks with periods $a_i$, the problem asks whether it is possible to assign each day to a task such that every task $i$ is scheduled at \textit{most} once every $a_i$ days. The density of an instance in either case is defined as the sum of the inverses of task periods. Recently, the long-standing $5/6$ density bound conjecture in the packing setting was resolved affirmatively. The resolution means any instance with density at least $5/6$ is schedulable. A corresponding conjecture was made in the covering setting and renewed multiple times in more recent work. We resolve this conjecture affirmatively by proving that every discretized point patrolling instance with density at least $\sum_{i = 0}^{\infty} 1/(2^i + 1) \approx 1.264$ is schedulable. This significantly improves upon the current best-known density bound of 1.546 and is, in fact, optimal. We also study the bamboo garden trimming problem, an optimization variant of the pinwheel problem. Specifically, given $n$ growth rates with values $h_i \in \mathbb{N}$, the objective is to minimize the maximum height of a bamboo garden with the corresponding growth rates, where we are allowed to trim one bamboo tree to height zero per time step. We achieve an efficient $9/7$-approximation algorithm for this problem, improving on the current best known approximation factor of $4/3$. 

\end{abstract}
\section{Introduction.}

Scheduling problems are a broad class of problems in computer science that involve assigning jobs to machines with specific time constraints. Applications range from industrial optimization to the entertainment industry to the functioning of operating systems \cite{applications}. 

The pinwheel problem is a well-studied periodic scheduling problem in which there is a single machine and a finite number of jobs that reappear immediately after being processed. For example, a vending machine may need to be restocked at least once every week, with the maintenance deadline depending on the last time it was restocked. Formalizing this, a pinwheel scheduling instance is a list of integers $(a_1, a_2, a_3, \cdots, a_n)$. The pinwheel decision problem asks whether a pinwheel scheduling instance has a valid schedule, in particular, a map $f: \mathbb{N} \rightarrow [n]$ corresponding to an assignment of days to jobs 1 through $n$ such that for all $i \in [n]$, every sequence of $a_i$ consecutive days has job $i$ assigned at least once.

The discretized point patrolling problem asks, for an instance $(a_1, a_2, a_3, \dots, a_n)$, whether a valid schedule exists, in particular, a map $f: \mathbb{N} \rightarrow [n]$ corresponding to an assignment of days to jobs such that for all $i \in [n]$, every sequence of $a_i$ consecutive days has job $i$ assigned at \textit{most} once. We also refer to this as the pinwheel covering problem, distinguishing it from the ``traditional'' pinwheel problem by qualifying the latter as the pinwheel packing problem for the remainder of this paper. 

Consider the fraction of all the days each job takes up for a pinwheel packing instance. It is at least $\frac{1}{a_i}$ for job $i$. We can formulate a notion of density of an instance as the sum of these fractions (i.e., $\sum_{i = 1}^n \frac{1}{a_i}$), and we can clearly see that if the density is greater than 1, no valid pinwheel schedule exists. However, density being less than or equal to 1 is not sufficient for schedulability. For example, the instance $(2, 3, 6)$ has density $\frac{1}{2} + \frac{1}{3} + \frac{1}{6} = 1$ but there is no valid schedule. Intuitively, the tasks with periods two and the three take up all the spaces, leaving no room for the task with period six. This example can be generalized.

\begin{fact}\label{fact:23theorem}

$(2, 3, a_3)$ is not schedulable for any $a_3 \in \mathbb{N}$ \cite{kawamura}.

\end{fact}

\cref{fact:23theorem} is proven in Appendix \ref{fact:23proof}, and applying the fact, we know that for all $k \in \mathbb{N}$, there exists a pinwheel scheduling instance that is not schedulable with density $\frac{1}{2} + \frac{1}{3} + \frac{1}{i} = \frac{5}{6} + \frac{1}{i}$. 

For pinwheel covering, we again define $D(A) = \sum_{a_i \in A} \frac{1}{a_i}$, where now if the density is less than 1, an instance cannot be schedulable. There is a corresponding version of \cref{fact:23theorem} in the covering setting.

\begin{fact}\label{fact:2itheorem}

$(2, 3, 5, 9, \cdots, 2^k + 1)$ is not schedulable for any $k \in \mathbb{N}$ (Theorem 17 of \cite{kawamura2020}).

\end{fact}

\cref{fact:2itheorem} is proven in \cref{fact:2iproof}, and applying the fact, we know that for all $i \in \mathbb{N}$, there exists a pinwheel covering instance that is not schedulable with density 
$$\sum_{n = 0}^{k} \frac{1}{2^n + 1} = \sum_{n = 0}^{\infty} \frac{1}{2^n + 1} - \sum_{n = k + 1}^{\infty} \frac{1}{2^n + 1} \geq -\sum_{n = k + 1}^{\infty} \frac{1}{2^n} + \sum_{n = 0}^{\infty} \frac{1}{2^n + 1} = -\frac{1}{2^k} + \sum_{n = 0}^{\infty} \frac{1}{2^n + 1}$$

Bamboo Garden Trimming (BGT) is a closely related optimization version of the pinwheel problem introduced by Gąsieniec \cite{bgtintro}. Specifically, we have a grove of $n$ bamboo plants, with growth rates $a_i \in \mathbb{N}$ (per day). We are allowed to trim one plant per day to height zero, and the objective value is the tallest plant in the grove across all time. A schedule achieves an objective of $H$ if and only if it is a valid schedule for the pinwheel packing instance $(\lfloor H/a_i \rfloor)_{i \in [n]}$ \cite{kawamura}.

In this paper, the phrase ``density bound'' refers to a density threshold that provides a schedulability guarantee. For the packing setting, a density bound of $d_p$ implies that every instance with density at most $d_p$ is schedulable. In the covering setting, a density bound of $d_c$ implies that every instance with density at least $d_c$ is schedulable. 

\subsection{Related Work.}

\cref{fact:23theorem} shows that a density bound of $\frac{5}{6}$ for the pinwheel packing problem would be optimal, and this was conjectured to be possible. There was significant work towards this goal, with bounds of $\frac{1}{2}$ \cite{05bound}, $\frac{2}{3}$ \cite{23bound}, $\frac{7}{10}$ \cite{710bound}, and $\frac{3}{4}$ \cite{34bound}. Finally, Kawamura $\cite{kawamura}$ proved the conjectured $\frac{5}{6}$ bound. One particularly relevant technique used in the $\frac{5}{6}$ proof is the concept of a fold operation, which enables us to reduce the problem to a finite number of instances and delegate the analysis to a computer \cite{towards}. We repurpose the folding idea (\cref{algo:fold}) to the covering setting and use the same folding operation of Kawamura \cite{kawamura} for bamboo garden trimming. 

\begin{algorithm}[ht]
  \DontPrintSemicolon
  \caption{Fold Operation $\pfold_{\theta}(A)$.}\label{algo:fold}
  Input: Pinwheel instance $A = (a_1, a_2, \dots, a_n)$ \;
  
  \While{$\max(A) > \theta$}{
    Let $a$ and $b$ be the largest and second largest values in $A$, allowing for repetition.
    
    \If{$b > \theta$}{    
    $A \gets A \ominus (a, b)$

    $A \gets A \sqcup (b/2)$
    }
    \Else{
    $A \gets A \ominus (a)$

    $A \gets A \sqcup (\theta)$
    }
    \Return $A$
}
\end{algorithm}

Note that we have used $A \sqcup (a)$ as the notation for adding a job of period $a$ to the instance $A$, and $A \ominus (a)$ for removing a job of period $a$ (with multiple periods in parenthesis indicating multiple additions or removals). The usefulness of \cref{algo:fold} relates to two properties of pinwheel packing instances, which can be used to show that $\pfold_{\theta}$ preserves instance unschedulability.

\begin{lemma}\label{lemma:packingmono}
If $A \sqcup (a)$ is schedulable, then $A \sqcup (b)$ is schedulable for any $b \geq a$ (Lemma 3 of \cite{kawamura}).
\end{lemma}

\begin{lemma}\label{lemma:packingpart}
If $A \sqcup (a)$ is schedulable, then $A \sqcup (\underbrace{a \cdot q, a \cdot q, \dots, a \cdot q}_{q})$ is schedulable (Lemma 3 of \cite{kawamura}).
\end{lemma}

\begin{lemma}

For all pinwheel packing instance $A$, if $A$ is unschedulable, then for all $\theta$, $\pfold_{\theta}(A)$ is unschedulable (Lemma 4 of \cite{kawamura}).

\end{lemma}

\cref{lemma:packingmono} is called the monotonicity property, and \cref{lemma:packingpart} is called the partitioning property. 

Just as $\frac{5}{6}$ is an optimal density bound for the packing setting, in the covering setting, due to \cref{fact:2itheorem}, if a density bound of $\sum_{i = 0}^{\infty} \frac{1}{2^i + 1}$ were true, it would be optimal. We refer to the conjecture that this density bound is indeed possible (i.e., every instance with density at least $\sum_{i = 0}^{\infty} \frac{1}{2^i + 1}$ is schedulable) as the $2^i + 1$ conjecture, and it was originally posed by Kawamura and Soejima (Conjecture 18 of \cite{kawamura2020}). This conjecture is particularly relevant due to the relatively recent resolution of the $\frac{5}{6}$ conjecture \cite{kawamura} after over 30 years. Currently, the best-known density bound is 1.546 \cite{kawamura2020}. 

The bamboo garden trimming problem \cite{bgtintro} can be considered a special type of cup game, which ties back to work done in the 1970s \cite{layland}. In particular, in the vanilla cup game, a filler repeatedly places $p$ units of water in some way across $n$ cups, and the emptier gets to choose $p$ cups to remove 1 unit of water from. Algorithms to achieve optimal backlog (water height) are known in the vanilla setting \cite{vanilla} as well as more complex settings \cite{variable, multip}. Bamboo Garden Trimming can be considered a form of cup game in the case of a constrained filler. There has been significant work on understanding guarantees for simple BGT strategies such as Reduce-Max (which trims the tree of the tallest height) and Reduce-Fastest (which trims the tree with the fastest growth rate among trees that have reached a certain height) \cite{priority, revisited}. Most work on the best approximation algorithms for bamboo garden trimming has been based on a reduction to pinwheel packing. There has been a series of improving approximation factors, from $2$ \cite{bgtintro}, to $\frac{32000}{16947}$ \cite{bgt188}, to $\frac{12}{7}$ \cite{bgt127}, to $1.6 + \epsilon$ \cite{gasieniec}, to $\frac{10}{7}$ \cite{bgt107}, and finally $\frac{4}{3}$ \cite{kawamura}. 

\vspace{5pt}

\subsection{Contributions.}

This paper presents two independent but related contributions. 

\begin{enumerate}
    \item This paper improves the state-of-the-art density bound for discretized point patrolling from 1.546 \cite{kawamura2020} to the optimal value of $\sum_{i = 0}^{\infty} \frac{1}{2^i + 1} \approx 1.264$. To do so, this paper extends the notions of monotonicity and partitioning to pinwheel covering and introduces two fold operations for pinwheel covering instances, which may be of independent interest. 
    \item This paper improves the state-of-the-art approximation factor among efficient algorithms for bamboo garden trimming from $\frac{4}{3} \approx 1.333$ \cite{kawamura} to $\frac{9}{7} \approx 1.286$. To do so, this paper leverages the existing framework of \cite{kawamura}, and we also utilize custom packing obstructions that arise due to both number-theoretic reasons and exhaustive search. For example, we prove that no instance which contains periods of 3 and 4 and density greater than $\frac{47}{48}$ can be scheduled. The parameter $\frac{47}{48}$ is usually replaced by the naive bound of 1 for general instances. To the best of our knowledge, our paper is the first to improve over the naive bound of 1 in an instance-dependent fashion. We believe that this technique can be combined with existing methods to produce significantly better bounds than $\frac{9}{7}$, although this may require a large amount of pre-computation time. 
\end{enumerate}

Several of our claims involve search over a finite number of cases, similar to that found in \cite{kawamura}. Supporting computational evidence for these claims is publicly available at: \url{https://github.com/ahanbm/optimal-dpp}

An independent and concurrent work by Kawamura and Kobayashi \cite{kawamura2025computerassistedproofoptimaldensity} matches our $\sum_{i = 0}^{\infty} \frac{1}{2^i + 1}$ density bound for pinwheel covering. 

\section{Technical Overview.}

In this section, we provide the high-level ideas and algorithms for our results in pinwheel covering and bamboo garden trimming.

\subsection{Pinwheel Covering.}\label{sec:pinover}

We start by setting up the basics for pinwheel covering. 

We define a fractional pinwheel covering problem, analogous to the fractional packing problem:

Consider a fractional instance $(a_1, a_2, a_3, \dots, a_n)$ with $a_i \in \mathbb{R^+}$ for all $i \in [n]$. We define a valid schedule as one in which days are assigned to jobs in a way so that for all $i \in [n]$ and for all $d \in \mathbb{N}$, every sequence of $d$ consecutive days has job $i$ scheduled at most $\lceil \frac{r}{a_i} \rceil$ times. We now formulate the analogue of the monotonicity and partitioning properties \cite{kawamura} in the covering setting.

\begin{lemma}\label{lemma:mono}
Monotonicity (Covering): If $A \sqcup (a)$ is schedulable, then $A \sqcup (b)$ is schedulable for any $b \leq a$. 
\end{lemma}

\begin{proof}
Given a valid schedule for $A \sqcup (a)$, we can reassign the days where the job with period $a$ is scheduled to the job with period $b$.
\end{proof}

\begin{lemma}\label{lemma:part}
Partitioning (Covering): If $A \sqcup (a)$ is schedulable, then $A \sqcup (\underbrace{a \cdot q, a \cdot q, \dots, a \cdot q}_{q \text{ times}})$ is schedulable.
\end{lemma}

\begin{proof}
Given a valid schedule for $A \sqcup (a)$, we can sequentially schedule the $q$ jobs of period $a \cdot q$ within the slots assigned to the job of period $a$, returning to the first job of period $a \cdot q$ after all are scheduled and repeating. Since there is at least a gap of $a$ between each scheduling of the original job of period $a$, there is a gap of at least $a \cdot q$ between repetitions of the same job of period $a \cdot q$ in the new instance.
\end{proof}

We now define a fold operation $\cfold_{\theta}(A)$, which is intended to take a pinwheel covering instance $A$ and transform it into an instance $A'$ such that no element of $A'$ is greater than $\theta \in \mathbb{N}$, and such that unschedulability is preserved. \cref{algo:naive} shows the major steps.

\vspace{5pt}

\begin{algorithm}[H]
\caption{Fold Operation $\cfold_{\theta}(A)$}\label{algo:naive}

Input: Pinwheel covering instance $A = (a_1, a_2, a_3, \dots, a_n)$

\While{\textnormal{the maximum value} $a \in A$ \textnormal{satisfies} $a > \theta$}{
    $A \gets A \ominus (a)$
    
    $b \gets \max(A)$

    $A \gets A \ominus (b)$

    \If{$a < 2b$}{
        $A \gets A \sqcup \left( \frac{a}{2} \right)$
    }

    \Else{
        $A \gets A \sqcup (b)$
    }
}
\Return{A}
\end{algorithm}

\vspace{5pt}

As in \cref{algo:fold}, we use $\ominus$ to represent removing a job and $\sqcup$ to represent adding one. \cref{algo:naive} can be regarded as the natural analogue of the fold operation for the packing setting introduced by Kawamura \cite{kawamura}. It can be shown that \cref{algo:naive} preserves unschedulability and decreases the density of $A$ by at most $\frac{2}{\theta}$. Similar to the fold operation in \cite{kawamura}, \cref{algo:naive} can be leveraged to prove a density bound of 1.3 for discretized point patrolling. Since \cref{algo:naive} preserves schedulability, we only need to consider instances with elements at most $\theta$ whose density is at least $1.3 - \frac{2}{\theta}$. Due to the fractional setting, the set of all possible instances with elements at most $\theta$ is infinite. We distill a finite set of instances by taking the ceiling of each element, and then leverage use computer analysis. A density bound of 1.3 is already a substantial improvement over the current state-of-the-art of 1.546, but proving the $2^i + 1$ conjecture requires more tools. 

We briefly stop to justify why the current fold operation is insufficient. The density loss for our fold operation is $\frac{2}{\theta}$ while the density loss for Kawamura's fold operation is $\frac{1}{\theta}$, so we might hope that the parameter 2 can be improved. However, it is tight, e.g., consider the instance $[\theta + 1, 2 \theta + 1, 4 \theta + 1, \dots, 2^n \cdot \theta + 1]$ for where $n$ is large and $\theta$ is a large power of 2. In order for the method described to work, we need to be able to schedule all instances of elements at most $\theta$ with density at least $-\frac{2}{\theta} + \sum_{i = 0}^{\infty} \frac{1}{2^i + 1}$. However, the instance $[2, 3, 5, \dots, 2^{\lceil \log_2{\theta} \rceil - 1} + 1]$ is within the desired density bound, and we have already shown it is not schedulable in \cref{fact:2itheorem}. 

\begin{fact}\label{fact:simtheorem}
    Let $A(\theta) = [2, 3, 5, \dots, 2^{\lceil \log_2{\theta} \rceil - 1} + 1]$. For any $\theta$, $D(A(\theta)) \geq -\frac{2}{\theta} + \sum_{i = 0}^{\infty} \frac{1}{2^i + 1}$.
\end{fact}

\cref{fact:simtheorem} is proved in \cref{fact:simproof}. This reasoning demonstrates that \cref{algo:fold} is insufficient to prove an optimal bound directly, so we develop an improved fold algorithm. 

\vspace{5pt}

\begin{algorithm}[H]
\caption{Improved Fold Operation $\cfold^{\text{imp}}_{\theta}(A)$}\label{algo:imp}

Input: Pinwheel covering instance $A = (a_1, a_2, a_3, \dots, a_n)$ and $\theta = 2^i$ for $i \in \mathbb{Z}_{>1}$

$A_{\text{max}} \gets \max(A)$

$\theta_{\text{current}} \gets 2^{\lceil \log_2(A_{\text{max}}) \rceil - 1}$

\While{$\theta_{\text{current}} \geq \theta$}{
    Let $n$ be the number of elements of $A$ in the range $(\theta_{\text{current}}, 2\theta_{\text{current}}]$

    \If{\textnormal{$n$ is odd and $n \ge 3$}}{
        Let $m_1 \leq m_2 \leq m_3$ be the three smallest elements of $A$ in the range $(\theta_{\text{current}}, 2\theta_{\text{current}}]$

        \If{$m_3 \le \frac{4}{3}\theta_{\text{current}}$}{
            $A \gets A \ominus (m_1, m_2, m_3)$

            $A \gets A \sqcup \left( \frac{m_3}{3} \right)$
        }
    }

    $A \gets \cfold_{\theta_{\text{current}}}(A)$

    $\theta_{\text{current}} \gets \theta_{\text{current}} / 2$
}
\Return{A}
\end{algorithm}

\cref{algo:imp} achieves a better guarantee than a density loss of $\frac{2}{\theta}$ in specific circumstances. In particular, it can be shown that for a single iteration, if there are at least two elements between $2^{i - 1}$ and $2^i$, the algorithm achieves a density loss bound of $\frac{3}{4 \theta}$. This bound is in contrast to the one iteration density loss bound for $\cfold_{\theta}$ of $\frac{1}{\theta}$. In addition, the improvement in density loss can be iterated under circumstances that we establish. 

This improvement enables us, with computer search, to determine that all pinwheel covering instances with density at least $\sum_{i = 0}^{\infty}\frac{1}{2^i + 1}$ are schedulable. The computer analysis is similar to that used for the 1.3 density bound but with more detailed casework. 

\subsection{Bamboo Garden Trimming.}\label{sec:bgtover}

We adopt the same reduction to pinwheel as \cite{kawamura2020}. Specifically, suppose we have an algorithm (efficient, in the sense of being a fast online scheduler \cite{kawamura}) $M$ that, for any pinwheel packing instance $(a_i)_{i \in [k]}$, either declares correctly that it is non-schedulable, or outputs a schedule for the ``relaxed” instance $(\lfloor \frac{9}{7} \cdot a_i \rfloor)_{i \in [k]}$. Then such an algorithm implies a $\frac{9}{7}$-approximation for BGT: 

Given a BGT instance $A = (a_i)_{i \in [k]}$, binary search for a height $H$ such that $M$ applied to $(\lfloor H/h_i \rfloor)_{i \in [k]}$ outputs a schedule but returns unschedulable if replaced by $H - 1$ (so the objective is no less than $H$). Since this schedule satisfies the pinwheel packing instance $(\lfloor \frac{9}{7} \cdot \lfloor H/h_i \rfloor \rfloor)_{i \in [k]}$, it satisfies $(\lfloor \frac{9}{7} \cdot H/h_i \rfloor)_{i \in [k]}$), and so it is achieves an objective of $\frac{9}{7} \cdot H$. We begin by providing an algorithm (\cref{algo:helper}) that will work in all but a few cases which can be handled recursively. 

\vspace{5pt}

\begin{algorithm}[H]
\DontPrintSemicolon
\caption{$M_{helper}$}\label{algo:helper}
Input: Pinwheel packing instance $A = (a_1, a_2, \dots, a_k)$ with $a_1 \leq a_2 \leq \dots \leq a_k$, and satisfying $(a_1) \neq (2)$, $(a_1, a_2) \neq (3, 3)$, and $(a_1, a_2, a_3, a_4) \neq (3, 6, 6, 6)$. $T_1$, $T_2$, and $T_3$ are pre-computed constant-sized lookup tables containing particular instances and their corresponding schedules.
  
\SetKwProg{Fn}{Function}{}{end}

\Fn{DB($A$)}{
    \If{$a_1 = 3$ \textnormal{and} $a_2 \in \{ 4, 5, 7 \}$}{
        \Return{$1 - \frac{1}{3a_2^2}$}\;
    }
    \ElseIf{$(a_1, a_2, a_3, a_4) = (3, 6, 6, 8)$}{
        \Return{$1 - \frac{1}{96}$}\;
    }
    \Else{
        \Return{$1$}\;
    }
}

\If{$D(A) > DB(A)$}{
\Return{\textnormal{``unschedulable''}}
}
  
\ElseIf{$a_1 = 3$}{
$A \gets \pfold_{18}(A)$

$A \gets (\lfloor \frac{9}{7} \cdot A_1 \rfloor, \lfloor \frac{9}{7} \cdot A_2 \rfloor, \dots, \lfloor \frac{9}{7} \cdot A_{k'} \rfloor)$

\vspace{1pt}

\If{$A \in T_1$}{
    \Return{$T_1[A]$}
}

\Else{
    $A \gets \pfold_{28}(A)$
    
    $A \gets (\lfloor \frac{9}{7} \cdot A_1 \rfloor, \lfloor \frac{9}{7} \cdot A_2 \rfloor, \dots, \lfloor \frac{9}{7} \cdot A_{k'} \rfloor)$
    
    \Return{$T_2[A]$} 
    }
}
\Else{
    $A \gets \pfold_{14}(A)$
    
    $A \gets (\lfloor \frac{9}{7} \cdot A_1 \rfloor, \lfloor \frac{9}{7} \cdot A_2 \rfloor, \dots, \lfloor \frac{9}{7} \cdot A_{k'} \rfloor)$
    
    \Return{$T_3[A]$}
}
\end{algorithm}

\vspace{5pt}

Note that $M_{helper}$ (\cref{algo:helper}) does not explicitly output the bamboo tree to trim on each day due to the intermediate use of the fold operation. However, the output can be regarded as an efficient representation of the schedule, as needed for a fast online scheduler \cite{kawamura}. We now have the tools to set $M$ (\cref{algo:bgt}). The recursive handling of $M$ can be considered a generalization of a technique used by \cite{kawamura} where it was only applied to job periods of 2. 

\begin{algorithm}[H]
\DontPrintSemicolon
\caption{$M$, a $\frac{9}{7}$-approximation for BGT.}\label{algo:bgt}
Input: Pinwheel instance $A = (a_1, a_2, \dots, a_k)$ with $a_1 \leq a_2 \leq \dots \leq a_k$.

  \If{$k = 1$}{
  \Return{\textnormal{the schedule that repeats task 1 every day}}.
  }

  \If{$a_1 = 2$}{    
  $A' \gets (\lfloor \frac{1}{2} \cdot a_i \rfloor)_{i \in [k] \setminus \{ 1 \}}$
  
  $S' \gets M \left( A' \right)$

  \If{$S' = \textnormal{unschedulable}$}{
  \Return{\textnormal{``unschedulable''}}
  }
  
  \Return{\textnormal{a schedule $S: \mathbb{Z} \rightarrow [k]$ defined by}
  $$S(t) = 
  \begin{cases}
      S'(t/2) + 1 & \text{for } t \equiv 0 \pmod{2} \\
      1 & \text{for } t \equiv 1 \pmod{2}      
  \end{cases}
  $$}
  }
  \ElseIf{$a_1 = 3$ \textnormal{and} $a_2 = 3$}{
  $A' \gets (\lfloor \frac{1}{3} \cdot a_i \rfloor)_{i \in [k] \setminus \{ 1, 2 \}}$
  
  $S' \gets M \left( A' \right)$

  \If{$S' = \textnormal{unschedulable}$}{
  \Return{\textnormal{``unschedulable''}}
  }
  
  \Return{\textnormal{a schedule $S: \mathbb{Z} \rightarrow [k]$ defined by}
  $$S(t) = 
  \begin{cases}
      S'(t/3) + 2 & \text{for } t \equiv 0 \pmod{3} \\
      1 & \text{for } t \equiv 1 \pmod{3} \\
      2 & \text{for } t \equiv 2 \pmod{3}
  \end{cases}
  $$}
  }
  \ElseIf{$a_1 = 3$, $a_2 = 6$, $a_3 = 6$ \textnormal{and} $a_4 = 6$}{
  $A' \gets (\lfloor \frac{1}{6} \cdot a_i \rfloor)_{i \in [k] \setminus \{ 1, 2, 3, 4 \}}$
  
  $S' \gets M \left( A' \right)$

  \If{$S' = \textnormal{unschedulable}$}{
  \Return{\textnormal{``unschedulable''}}
  }
  
  \Return{\textnormal{a schedule $S: \mathbb{Z} \rightarrow [k]$ defined by}
  $$S(t) = 
  \begin{cases}
      S'(t/6) + 4 & \text{for } t \equiv 0 \pmod{6} \\
      1 & \text{for } t \equiv 1, 4 \pmod{6} \\
      2 & \text{for } t \equiv 2 \pmod{6} \\
      3 & \text{for } t \equiv 3 \pmod{6} \\
      4 & \text{for } t \equiv 5 \pmod{6}
  \end{cases}
  $$}
  }
  \Else{
  \Return{$M_{helper}(A)$}
  }
\end{algorithm}

\section{An Improved Density Bound for Pinwheel Covering.}\label{sec:13}

In this section, we prove a 1.3 density bound for pinwheel covering. We rely on \cref{algo:naive} along with a new program for pinwheel covering adapting some computer search ideas from the packing setting \cite{towards}. We start by proving properties of \cref{algo:naive}.

\begin{lemma}\label{claim2}

Every element of $A' = \cfold_{\theta}(A)$ is at most $\theta$.

\end{lemma}

\begin{proof}
We can clearly see from the while loop condition in \cref{algo:naive} that upon termination, the maximum value of the returned $A$ must be no greater than $\theta$. Furthermore, the algorithm terminates because, at every iteration of the while loop, the length of $A$ decreases by 1; thus, there are at most $n$ iterations.
\end{proof}

\begin{lemma}\label{claim3}
If $A$ is unschedulable, then $\cfold_\theta(A)$ is unschedulable.
\end{lemma}

\begin{proof}
We prove the contrapositive: if $\cfold_\theta(A)$ is schedulable, then $A$ is schedulable. 

We prove the claim by induction on the number of remaining while loop iterations. For the base case, when there are zero iterations left, the returned value is assumed to be schedulable. If the instance $A_r$, the value of $A$ with $r$ remaining while loop iterations, is schedulable, then we need to prove $A_{r + 1}$ is also schedulable. There are two cases: either the $a$ and $b$ (highest and second-highest values) in $A_{r + 1}$ were replaced by $\left( \frac{a}{2} \right)$ (in the case $a < 2b$) or by $(b)$. In the case of $(b)$, $A_{r + 1}$ is the same as $A_r$ but with an extra job, so $A_{r + 1}$ is schedulable according to the monotonicity property (we can treat $A_r$ as having an $\infty$ period job). 

In the case where $a < 2b$ and the $\left( \frac{a}{2} \right)$ job in $A_r$ arose from the $a$ and $b$ jobs in $A_{r + 1}$, we know that schedulability of $(A_r \ominus \left( \frac{a}{2} \right)) \sqcup \left( \frac{a}{2} \right)$ implies schedulability of $(A_r \ominus \left( \frac{a}{2} \right)) \sqcup \left( a, a \right)$ (by the partitioning property) and this in turn implies schedulability of $(A_r \ominus \left( \frac{a}{2} \right)) \sqcup \left( a, b \right) = A_{r + 1}$ (by the montonicity property). Since we covered both possible cases, this completes the inductive step. Inducting till $A_{n'}$ where $n'$ is the total number of while loop iterations to evaluate $\cfold_\theta(A)$, we have that $A$ is schedulable, as desired.
\end{proof}

\begin{lemma}\label{claim4}
$D(\textnormal{cfold}_\theta(A)) \geq D(A) - \frac{2}{\theta}$
\end{lemma}

\begin{proof}
We split the density loss analysis into two parts: the last while loop iteration and the rest. For any iteration, we claim that the density loss is at most $\frac{1}{b} - \frac{1}{a}$, where $a$ and $b$ are the maximum and second-maximum values before that iteration. In particular, if $a$ is deleted, then $a \geq 2b$ so $b \leq \frac{a}{2}$ and the density loss is $\frac{1}{a} = \frac{2}{a} - \frac{1}{a} \leq \frac{1}{b} - \frac{1}{a}$. If $(a, b)$ is replaced with $\left( \frac{a}{2} \right)$ then the density loss is $\frac{1}{a} + \frac{1}{b} - \frac{2}{a} = \frac{1}{b} - \frac{1}{a}$, with direct equality. The total density loss for all iterations but the last is

$$\frac{1}{b_1} - \frac{1}{a_1} + \frac{1}{b_2} - \frac{1}{a_2} + \frac{1}{b_3} - \frac{1}{a_3} + \dots + \frac{1}{b_{x - 1}} - \frac{1}{a_{x - 1}}$$

where we have said that $a_i$ and $b_i$ are the maximum and second-maximum values before the $i$th iteration, and there are $x$ total iterations. Furthermore, $b_x \geq a_{x + 1}$ for all $x$, since at each iteration $a$ either gets deleted or replaced with $\frac{a}{2}$ in the case that $\frac{a}{2} < b$. Therefore, the sum above can be upper-bounded by a telescoping sum:

$$\frac{1}{b_{x - 1}} - \frac{1}{a_{x - 1}} + \frac{1}{a_{x - 1}} - \frac{1}{a_{x - 2}} + \dots + \frac{1}{a_2} - \frac{1}{a_1} = \frac{1}{b_{x - 1}} - \frac{1}{a_1} \leq \frac{1}{b_{x - 1}}$$

Since there is another iteration and $a_x \leq b_{x - 1}$, we must have that $b_{x - 1} \geq \theta$ so the total density loss in all iterations but the last is at most $1/\theta$. For the last iteration, we can make a similar claim as before and say that the density loss is at most $1/a$ at any iteration, by similar logic as before. Either $a$ is deleted for exact loss of $\frac{1}{a}$ or the loss is $\frac{1}{a} + \frac{1}{b} - \frac{2}{a} = \frac{1}{b} - \frac{1}{a}$ and this is in the case $a < 2b$ so $\frac{1}{b} - \frac{1}{a} < \frac{2}{a} - \frac{1}{a} = \frac{1}{a}$. Therefore, the loss in the last iteration is at most $1/\theta$, for a total density loss bound of 
$$\frac{1}{\theta} + \frac{1}{\theta} = \frac{2}{\theta}$$

as desired.
\end{proof}

\begin{lemma}\label{claim5}
Any instance $A = (a_1, a_2, \dots, a_n)$ whose periods are integers at most 16 and satisfies $D'(A) \geq 1.3 - \frac{2}{16}$ is schedulable, where

\[
D'(A) = \sum_{i = 1}^n
\begin{cases} 
1/a_i & \text{if } a_i \leq 8, \\
1/(a_i - 1) & \text{if } a_1 > 8.
\end{cases}
\]
\end{lemma}

\begin{proof}
The set of instances of such $A$ is infinite, but we can consider the restricted subset $R$ that contains only those $A$ that would be underneath the density bound if another element were to be removed. Specifically, to convert a general $A$ into an element of $R$, we remove the maximum element of $A$ until another removal would cause $A_{removed}$ to go under the density bound of $1.3 - \frac{2}{16}$. By monotonicity, if every element of $R$ is schedulable, then \cref{claim5} is true. Since $R$ contains a finite number of elements, it suffices to compute a schedule for each possibility, which we have done.
\end{proof}

\begin{theorem}\label{claim6}

Every pinwheel covering instance with density at least 1.3 is schedulable. 

\end{theorem}

\begin{proof}
Suppose, for the sake of contradiction, that there exists an unschedulable instance $A$ with $D(A) \geq 1.3$. Then, $\cfold_{16}(A)$ is also unschedulable, according to \cref{claim3}. In addition, $D(\cfold_{16}(A)) \geq D(A) - \frac{2}{16}$. Consider $A'$ which contains $a$ for all $a \in \cfold_{16}(A)$ such that $a \leq 8$, and additionally contains $\lceil a \rceil$ for all elements $a \in \cfold_{16}(A)$ such that $a > 8$, with appropriate multiplicity. Since $\cfold_{16}(A)$ is unschedulable, $A'$ is also unschedulable according to the monotonicity property. In addition, since the elements of $\cfold_{16}(A)$ that are at most 8 are unchanged, and the rest increase by at most 1, $D'(A') \geq D(\cfold_{16}(A))$, implying that $D'(A') \geq D(A) - \frac{2}{16}$ and since $D(A) \geq 1.3$, $D'(A') \geq 1.3 - \frac{2}{16}$. Furthermore, because $\cfold_{16}(A)$ produces elements that are at most 16 according to \cref{claim2}, $A'$ does not increase elements above 16, and every element of $A'$ is an integer, \cref{claim5} applies. In particular, $A'$ is schedulable, a contradiction. This proves a bound of 1.3 for pinwheel covering.
\end{proof}

\section{An Optimal Density Bound for Discretized Point Patrolling.}\label{sec:opt}

In this section, we establish an optimal density bound for discretized point patrolling. Specifically, we prove that any instance $L$ with density at least
$\sum_{i = 0}^{\infty} \frac{1}{2^i + 1}$ is schedulable.

\begin{lemma}\label{claim8}

Consider any instance $L$ such that $D(L) \geq \sum_{i = 0}^{\infty} \frac{1}{2^i + 1}$ and let $L \leq r$ be the set of jobs in $L$ restricted to those with period less than $r$. Let $W = [2, 3, 5, \dots, 2^i + 1, \dots]$. Note that $D(W) = \sum_{i = 0}^{\infty} \frac{1}{2^i + 1}$, so $D(L) \geq D(W)$. Let $S$ be the set of powers of 2, $p$, such that the following is true:
$$D(L \leq p) > D(W \leq p)$$

Then, $S$ is not empty.

\end{lemma}

\begin{proof}
Let $L_{\text{max}}$ be the maximum element of $L$. Then, we claim that setting $p = 2^{\lceil \log_2(L_{\text{max}}) \rceil}$ satisfies the inequality. Note that $2^{\lceil \log_2(L_{\text{max}}) \rceil} \geq L_{\text{max}}$, so 
$$D(L \leq p) = D(L)$$

In addition, $D(W) > D(W \leq p)$ because $W$ contains elements above $p$ due to having unboundedly large elements. We have that $D(L) \geq D(W)$ from before, so putting everything together, we have 
$$D(L \leq p) = D(L) \geq D(W) > D(W \leq p)$$

so $D(L \leq p) > D(W \leq p)$. This is the exact criterion for being included in $S$, proving the claim.
\end{proof}

\begin{lemma}\label{claim9}

Consider any instance $L$ such that $D(L) \geq \sum_{i = 0}^{\infty} \frac{1}{2^i + 1}$. Since $S$, as defined in \cref{claim8}, is a non-empty set of integers, by the well-ordering principle, it has a minimum element, call it $E$. Let $N_i(A)$ be the number of elements of a pinwheel covering instance $A$ in the range $(i, 2i]$.

Then, $N_{E/2}(L) \geq 2$.

\end{lemma}

\begin{proof}
Since $E$ is contained in $S$, it must be that 
$$D(L \leq E) > D(W \leq E)$$

In addition, since $E$ is the minimum element of $S$, it must be that $E/2$ is not an element of $S$, meaning
$$D(L \leq E/2) \leq D(W \leq E)$$

Since $D(L \leq E) - D(W \leq E) > 0$ and $D(L \leq E/2) - D(W \leq E/2) \leq 0$, we have
$$\left[ D(L \leq E) - D(W \leq E) \right] - \left[ D(L \leq E/2) - D(W \leq E/2) \right] > 0$$

Simplifying,
$$D(L \leq E) - D(L \leq E/2) > D(W \leq E) - D(W \leq E/2)$$

Note that since $L$ is composed of integers, the left-hand side is the density of elements in the range $(E/2, E]$. In addition, since $W$ is composed of $2^i + 1$ for every natural number $i$ and $E/2$ is a power of 2, $E/2 + 1 \in W$, and $D(W \leq E) - D(W \leq E/2) = \frac{1}{E/2 + 1}$. Therefore, 
$$D(L \leq E) - D(L \leq E/2) > \frac{1}{E/2 + 1}$$

Since the maximum possible density of a single integer in the range $(E/2, E]$ is $\frac{1}{E/2 + 1}$, having density strictly greater than $\frac{1}{E/2 + 1}$ in that range means that $L$ contains at least two elements in the range $(E/2, E]$. 
\end{proof}

We can now analyze the improved fold algorithm outlined in \cref{algo:imp}.

\begin{lemma}\label{claim10}

For any instance $A$, every element of $A' = \textnormal{cfold}_{\theta}^{\text{imp}}(A)$ is at most $\theta$.

\end{lemma}

\begin{proof}
We can see that to get out of the while loop of line 4, $\theta_{\text{current}}$ must be less than $\theta$. In addition, since both are powers of two at all times, $\theta_{\text{current}}$ must be $\theta/2$ after the while loop exit, so there is a while loop iteration with $\theta_{\text{current}} = \theta$. Now, noting line 11 and \cref{claim2}, before going to line 12, every element of $A$ must be no greater than $\theta$. Line 11 always terminates because there are a finite number of elements greater than $\theta_{\text{current}}$ at each iteration, and at least one such element is removed per iteration of \cref{algo:naive}. After the final execution of line 11, the elements are not increased, so the returned instance has elements at most $\theta$.
\end{proof}

\begin{lemma}\label{claim11}

For any instance $A$, if $\textnormal{cfold}_{\theta}^{\text{imp}}(A)$ is schedulable, then $A$ is schedulable. 

\end{lemma}

\begin{proof}
We prove the claim by induction on the remaining iterations of the while loop in line 5 of \cref{algo:imp}. Let $A_i$ be the value of $A$ after the $i$th iteration of the while loop. If $n$ is the total number of iterations, we can see that $A_n = \text{cfold}_{\theta}^{\text{imp}}(A)$ so the claim is true in the base case. We want to prove that if $A_{t + 1}$ is schedulable, then $A_{t}$ is schedulable. Let $A_{(t + 1)b}$ be the value of $A$ after line 6 of the $(t + 1)$-th iteration of line 4 and let $A_{(t + 1)c}$ be the value of $A$ after line 11 of the $(t + 1)$-th iteration of line 4 (this is equivalent to $A_{t + 1}$).

According to \cref{claim3}, if $A_{(t + 1)c}$ is schedulable then $A_{(t + 1)b}$ is schedulable. To prove that $A_{t}$ is schedulable, we do casework based on whether the if condition of line 6 is true or false. If false, then $A_t = A_{(t + 1)b}$ so $A_t$ is schedulable. If true, then we can see schedulability by imagining the process in two steps, one in which $(m_1, m_2, m_3)$ is replaced by $(m_3, m_3, m_3)$ and another in which $(m_3, m_3, m_3)$ is replaced by $\left( \frac{m_3}{3} \right)$. The first step preserves unschedulability due to the monotonicity property (\cref{lemma:mono}), and the second step preserves unschedulability due to the partitioning property (\cref{lemma:part}). By the principle of mathematical induction, $A$ is schedulable. 
\end{proof}

\begin{lemma}\label{claim12}
Let us define $A_i$ as the value of $A$ after the $i$th iteration of line 4 of \cref{algo:imp} and $\theta_i$ as the value of $\theta_{\text{current}}$ during the $i$th iteration of the while loop on line 4 (so $\theta_1$ is the initial value). 

For any instance $A$ and any $i \geq 1$, if $N_{\theta_{i}}(A_{i - 1}) \geq 2$, then
$$D(A_{i}) \geq D(A_{i - 1}) - \frac{3}{4 \theta_{i}}$$

In addition, if $N_{\theta_i}(A_{i - 1})$ is even, then
$$D(A_{i}) \geq D(A_{i - 1}) - \frac{1}{2 \theta_{i}}$$
\end{lemma}

\cref{claim12} is proved in Appendix \ref{proof:claim12}. Note that the naive density loss bound is $\frac{1}{\theta_i}$ even without any requirement on $N_{\theta_i}(A_{i - 1})$. The improved $\frac{1}{2 \theta_i}$ bound holds even for the initial fold operation (\cref{algo:naive}) for cases in which $N_{\theta_i}(A_{i - 1})$ is even. The improved $\frac{3}{4 \theta_i}$ bound in the odd case arises from our thirding operation and optimizing the parameter $\frac{4}{3}$ in line 8 of \cref{algo:imp}. 

\begin{lemma}\label{claim13}

For any instance $A$ with density at least $\sum_{i = 0}^{\infty} \frac{1}{2^i + 1}$, if $\theta \geq 16$, $N_{\theta_1}(A) \geq 2$, and $2 \theta_1$ is the minimum value of $p$ that is a power of 2 satisfying $D(A \leq p) > D(W \leq p)$, then
$$D(\textnormal{cfold}_{\theta}^{\text{imp}}(A)) \geq D(A) - \sum_{m = 1}^{1 + \log_2\left( \theta_1 / \theta \right)} \frac{3}{4 \theta_m}$$

\end{lemma}

\cref{claim13} is proved in Appendix \ref{proof:claim13}. The high-level idea is to iterate the density loss bound from \cref{claim12} by arguing that if $N_{\theta_i}(A_{i - 1}) \geq 2$ holds initially, then it continues to hold.

\begin{theorem}\label{claim14}

Any instance $L$ with density at least $\sum_{i = 0}^{\infty} \frac{1}{2^i + 1}$ is schedulable. 

\end{theorem}

\cref{claim14} is proved in Appendix \ref{proof:claim14}. Given \cref{algo:imp}, the idea is similar to the proof of \cref{claim5}, but involving more cases and subcases. In particular, unlike \cref{claim5}, there is a parameter $E$ (from \cref{claim9}) which is relevant to our analysis, so we must do casework based on $E$. Furthermore, \cref{claim14} is exactly the $2^i + 1$ conjecture, proving the desired optimal density bound.
 
\section{An Improved Approximation for BGT.}\label{sec:bgt}

We begin by proving that $M_{helper}$ (\cref{algo:helper}) provides correct output assuming that the input satisfies the preconditions. That is, we prove that if the algorithm outputs ``unschedulable'' then $A$ is really unschedulable and if it outputs a schedule, then it is a valid schedule for $(\lfloor \frac{9}{7} \cdot A_1 \rfloor, \lfloor \frac{9}{7} \cdot A_2 \rfloor, \dots, \lfloor \frac{9}{7} \cdot A_k \rfloor)$. We begin with the unschedulable case. For the $DB(A) = 1$ case, it can clearly be seen that the output is correct because it only declares schedulability if $D(A) > 1$. We prove the case of $DB(A) = 1 - \frac{1}{3a_2^2}$ in a more general way.

\begin{lemma}\label{lemma:bgt1}

For any $a, b \in \mathbb{N}_{>1}$ with $a < b$ such that $\gcd(a, b) = 1$, a pinwheel instance $A$ satisfying $(a_1, a_2) = (a, b)$ and $D(A) > 1 - \frac{1}{ab^2}$ is unschedulable. 

\end{lemma}

\cref{lemma:bgt1} is proved in Appendix \ref{proof:bgt1}. We can prove the $\frac{95}{96}$ bound for the $(3, 6, 6, 8)$ case similarly.

\begin{lemma}\label{lemma:bgt2}

Any pinwheel instance $A$ satisfying $(a_1, a_2, a_3, a_4) = (3, 6, 6, 8)$ and $D(A) > 1 - \frac{1}{96}$ is unschedulable. 

\end{lemma}

\cref{lemma:bgt2} is proved in Appendix \ref{proof:bgt2}. We now prove that if $M_{helper}$ (\cref{algo:helper}) outputs a schedule, then the schedule is valid. 

\begin{lemma}\label{lemma:bgt3}

For any pinwheel instance $A$ and any $\theta \in \mathbb{N}$, let $A_f = (\lfloor \frac{9}{7} \cdot A_1 \rfloor, \lfloor \frac{9}{7} \cdot A_2 \rfloor, \dots, \lfloor \frac{9}{7} \cdot A_{k} \rfloor)$, let $B = \pfold_{\theta}(A)$ and let $C = (\lfloor \frac{9}{7} \cdot B_1 \rfloor, \lfloor \frac{9}{7} \cdot B_2 \rfloor, \dots, \lfloor \frac{9}{7} \cdot B_{k'} \rfloor)$. 

If $C$ is schedulable, then $A_f$ is schedulable.

\end{lemma}

\cref{lemma:bgt3} is proved in \cref{proof:bgt3}. Now that we have proven the correctness of $M_{helper}$, we focus on proving the correctness of $M$ (\cref{algo:bgt}) under the assumption that $M_{helper}$ (\cref{algo:helper}) is correct. 

\begin{theorem}

$M$ is a $\frac{9}{7}$-approximation algorithm for Bamboo Garden Trimming.

\end{theorem}

\begin{proof}
We prove the correctness of $M$ (\cref{algo:bgt}) by induction on the number of jobs in $A$. The output is clearly correct in the base case (else case) because $M$ directly delegates to $M_{helper}$ and satisfies its preconditions due to filtering out the special cases. The $a_1 = 2$ case is directly given in \cite{kawamura2020}. Specifically, declaring $A$ unschedulable if $A'$ is unschedulable is justified because if $A$ were schedulable, all the instances of task 1 could be removed to generate a schedule for $A'$. Now, if a schedule is found for $A'$ (which is correct, by induction), then job $i > 1$ is scheduled at least every $\lfloor \frac{9}{7} \cdot \lfloor  \frac{1}{2} \cdot a_i \rfloor \rfloor$ days, so in the final schedule $S$ it is scheduled at least every $2 \cdot \lfloor \frac{9}{7} \cdot \lfloor \frac{1}{2} \cdot a_i \rfloor \rfloor \leq \lfloor \frac{9}{7} \cdot a_i \rfloor$ days, as desired.

For the $a_1 = 3, a_2 = 3$ case, we can declare $A$ unschedulable if $A'$ is unschedulable, since if $A$ were schedulable, all instances of tasks 1 and 2 could be removed to generate a schedule for $A$, since in every 3 days there are at least two combined schedulings of jobs 1 and 2. If a schedule is found for $A'$, then job $i > 2$ is scheduled at least every $\lfloor \frac{9}{7} \cdot \lfloor \frac{1}{3} \cdot a_i \rfloor \rfloor$ days, so in the final schedule $S$ it is scheduled at least every $3 \cdot \lfloor \frac{9}{7} \cdot \lfloor \frac{1}{3} \cdot a_i \rfloor \rfloor \leq \lfloor \frac{9}{7} \cdot a_i \rfloor$ days, as desired. For the $a_1 = 3, a_2 = 6, a_3 = 6, a_4 = 6$ case, we can declare $A$ unschedulable if $A'$ is unschedulable, since if $A$ were schedulable all instances of tasks 1, 2, 3, and 4 could be removed to generate a schedule for $A$, since in every 6 days there are at least five combined schedulings of jobs 1, 2, 3, and 4. If a schedule is found for $A'$, then job $i > 4$ is scheduled at least every $\lfloor \frac{9}{7} \cdot \lfloor \frac{1}{6} \cdot a_i \rfloor \rfloor$ days, so in the final schedule $S$ it is scheduled at least every $6 \cdot \lfloor \frac{9}{7} \cdot \lfloor \frac{1}{6} \cdot a_i \rfloor \rfloor \leq \lfloor \frac{9}{7} \cdot a_i \rfloor$ days, as desired. This completes the proof of correctness of $M$.
\end{proof}

\section{Conclusion.}

In this paper, we have proven that for any discretized point patrolling instance $A$ with $D(A) \geq \sum_{n = 1}^{\infty} \frac{1}{2^n + 1}$, $A$ is schedulable. This resolves a conjecture posed in several related works \cite{kawamura, kawamura2020, covering} and complements the recent proof of the $\frac{5}{6}$ density bound in the pinwheel packing setting \cite{kawamura}.

Our bound is provably optimal, so it puts the question of the right density bound for general discretized point patrolling instances to rest. However, there has been recent work on the density bound for instances depending on the minimum element in the instance \cite{covering, gasieniec, regular}. The right density bound is not known in this case, even asymptotically. Besides discretized point patrolling, we have also produced an algorithm that achieves a state-of-the-art $\frac{9}{7}$-approximation for the bamboo garden trimming problem. With respect to future research, we field the following conjecture: 

\begin{conjecture}\label{conj:bgt}
    For any $r > 1$, there is an efficient $r$-approximation for Bamboo Garden Trimming.
\end{conjecture}

To the best of our knowledge, our work has been the first to establish instance-dependent density bounds beyond the naive bound of 1 for bamboo garden trimming instances. We hope that this, in conjunction with the fold operation and exhaustive computer search ideas introduced by Kawamura \cite{kawamura}, can produce a decreasing sequence of bounds towards \cref{conj:bgt}.

\appendix

\section{Proof of Facts.}

\subsection{\texorpdfstring{Proof of \cref{fact:23theorem}.}{Proof of Fact 1.}}\label{fact:23proof}

\begin{proof}
Suppose, for the sake of contradiction, that a valid schedule exists. Let $i$ be a day on which job $3$ is scheduled such that $i > 1$. Consider the jobs assigned on days $i - 1$ and $i + 1$. Suppose, for the sake of contradiction, that day $i - 1$ is not assigned to job 1. Then, days $i - 1$ and $i$ are a sequence of length 2 with no days assigned to job 1 but $a_1 = 2$, so this contradicts the definition of a valid schedule. Therefore, day $i - 1$ must be assigned to job 1. Similarly, day $i + 1$ must be assigned to job 1. However, then day $i - 1$, $i$ and $i + 1$ form a sequence of length 3 with no days assigned to job 2 but $a_2 = 3$, a contradiction. So, we have proven that $(2, 3, a_3)$ is not schedulable.
\end{proof}

\subsection{\texorpdfstring{Proof of \cref{fact:2itheorem}.}{Proof of Fact 2.}}\label{fact:2iproof}

\begin{proof}
We restate the proof given by Kawamura and Soejima \cite{kawamura2020}, proving the claim by induction on $k$, attempting to prove that with $k$ jobs we cannot find a valid schedule for even $2^k$ consecutive days. For $k = 1$, we have a single job of period 2 and want to schedule the first two days. Since only one of the two can be assigned to the single job, the base case is true. 

For the inductive case, suppose we have some $k = n + 1$. If the job of period $2^{n + 1} + 1$ is scheduled past day $2^n$, then the first $n$ jobs are validly assigned to the first $2^n$ days. Otherwise, the job of period $2^{n + 1} + 1$ is scheduled before or at day $2^n$, in which case the first $n$ jobs are validly assigned on the last $2^n$ days. Either way, we contradict the inductive hypothesis.
\end{proof}

\subsection{\texorpdfstring{Proof of \cref{fact:simtheorem}.}{Proof of Fact 3.}}\label{fact:simproof}

\begin{proof}
Let $n = \lceil \log_2 \theta \rceil$. We have

$$D(A) = \left( \sum_{i = 0}^{\infty} \frac{1}{2^i + 1} \right) - \frac{1}{2^n + 1} + \frac{1}{2^{n + 1} + 1} + \dots = \sum_{i = 0}^{\infty} \frac{1}{2^i + 1} - \sum_{i = n}^{\infty} \frac{1}{2^i + 1}$$
$$\geq \sum_{i = 0}^{\infty} \frac{1}{2^i + 1} - \sum_{i = n}^{\infty} \frac{1}{2^i} =  -\frac{2}{2^n} + \sum_{i = 0}^{\infty} \frac{1}{2^i + 1}$$

Now, note that since $n = \lceil \log_2 \theta \rceil$, $2^n \geq \theta$. Therefore, 
$$D(A) \geq -\frac{2}{2^n} + \sum_{i = 0}^{\infty} \frac{1}{2^i + 1} \geq -\frac{2}{\theta} + \sum_{i = 0}^{\infty} \frac{1}{2^i + 1}$$

as desired.
\end{proof}

\section{\texorpdfstring{Omitted Proofs in Section \ref{sec:opt}.}{Omitted Proofs in Section 4.}}

\subsection{\texorpdfstring{Proof of \cref{claim12}.}{Proof of Lemma 2.5.}}\label{proof:claim12}

\begin{proof}
If $n = 0$ for the even case, we will have zero loss and $0 \leq \frac{1}{2 \theta_i}$. Now, since $N_{\theta_{i}}(A_{i - 1}) \geq 2$, if $n$ (from line 6 of \cref{algo:imp}) is odd, then the if condition will be triggered; otherwise, if $n$ is even, it will not. We split these up with casework.

\vspace{10pt}

\textbf{Case 1}: $n$ is even

The if condition on line 6 of \cref{algo:naive} (which is activated from line 11 of \cref{algo:imp}) is always true since both $a$ and $b$ are in the range $(\theta_i, 2 \theta_i]$. This is because there are an even number of values in this range initially, and we remove two at a time (after dividing a number in this range by 2 it will be outside of the range). Let $(a_j, b_j)$ be the values of $a$ and $b$ during the $j$th iteration of the while loop of \cref{algo:naive}. Note that the loss on iteration $j$ is
$$\frac{1}{b_j} - \frac{1}{a_j}$$

because we can treat line 7 of \cref{algo:naive} as increasing the period of a job from $b_j$ to $a_j$, which induces the density loss above, and then dividing by 2 according to partitioning, which induces no density loss. 

In addition, because $b_j$ is the second-largest element of $A$ during iteration $j$ and both $a_j$ and $b_j$ are removed, we must have $a_{j + 1} \leq b_j$. Therefore, the total loss is
$$\sum_{j = 1}^{n/2} \left( \frac{1}{b_j} - \frac{1}{a_j} \right) \leq \sum_{j = 1}^{n/2 - 1} \left( \frac{1}{a_{j + 1}} - \frac{1}{a_j} \right) + \frac{1}{b_{n/2}} - \frac{1}{a_{n/2}}$$

Noting that the sum on the right hand side is telescoping, we can simplify the expression to
$$\frac{1}{a_{n/2}} - \frac{1}{a_1} + \frac{1}{b_{n/2}} - \frac{1}{a_{n/2}} = \frac{1}{b_{n/2}} - \frac{1}{a_1}$$

We have already discussed that $a, b \in (\theta_i, 2 \theta_i]$ for all iterations, so $b_{n/2} > \theta_i$ and $a \leq 2 \theta_i$. Thus,
$$\frac{1}{b_{n/2}} - \frac{1}{a_1} \leq \frac{1}{\theta_i} - \frac{1}{2 \theta_i} = \frac{1}{2 \theta_i}$$

Now, $\frac{1}{2 \theta_i} \leq \frac{3}{4 \theta_i}$, as desired.

\vspace{10pt}

\textbf{Case 2}: $n$ is odd

Recall that the condition of line 6 of \cref{algo:imp} is asserted in this case, so we split the casework depending on whether the condition of line 8 is true or false

\vspace{10pt}

\textbf{Subcase 1}: $m_3 \leq \frac{4}{3} \theta_i$

The density loss from the combined effect of lines 9 and 10 is
$$\frac{1}{m_1} - \frac{1}{m_3} + \frac{1}{m_2} - \frac{1}{m_3}$$

Since both $m_1$ and $m_2$ are defined to be in the range $(\theta_i, 2 \theta_i]$, this density loss is at most
$$\frac{2}{\theta_i} - \frac{2}{m_3}$$

To bound the additional density loss from line 11 of \cref{algo:imp}, we renew our analysis of case 1. In particular, since we started with an odd number of elements and removed 3, we must now have an even number of remaining elements in the range $(\theta_i, 2 \theta_i]$, specifically $n - 3$. If we define $(a_j, b_j)$ in the same way as before, we have that the total density loss attributable to line 11 is at most
$$\frac{1}{b_{(n - 3)/2}} - \frac{1}{a_1}$$

As before, $a_1 \leq 2 \theta_i$. In addition, because $m_1$, $m_2$, and $m_3$ were defined as the three smallest elements of $A$ in the range $(\theta_i, 2 \theta_i]$, we have $b_{(n - 3)/2} \geq m_3$. Therefore,
$$\frac{1}{b_{(n - 3)/2}} - \frac{1}{a_1} \leq \frac{1}{m_3} - \frac{1}{2 \theta_i}$$

Now, the total loss on iteration $i$ of line 4 is upper-bounded by
$$\left( \frac{2}{\theta_i} - \frac{2}{m_3} \right) + \left( \frac{1}{m_3} - \frac{1}{2 \theta_i}\right) = \frac{3}{2 \theta_i} - \frac{1}{m_3}$$

Recalling that $m_3 \leq \frac{4}{3} \theta_i$ we have
$$\frac{3}{2 \theta_i} - \frac{1}{m_3} \leq \frac{3}{2 \theta_i} - \frac{3}{4 \theta_i} = \frac{3}{4 \theta_i}$$

This proves an overall density loss of at most $\frac{3}{4 \theta_i}$ for this subcase, as desired.

\textbf{Subcase 2}: $m_3 > \frac{4}{3} \theta_i$

We again renew the analysis of case 1, getting that the density loss on the first $(n - 1)/2$ iterations (recalling that $n \geq 3$) is at most
$$\frac{1}{b_{(n - 1)/2}} - \frac{1}{a_1}$$

Now, note that on the first iteration of line 2 of \cref{algo:naive} that occurs during line 11 of \cref{algo:imp}, $(a_1, b_1)$ will be replaced with $\left( \frac{2}{a_1} \right)$ in $A$ because both are in the range $(\theta_1, 2 \theta_1]$ so it cannot be that $a_1 \geq 2b_1$. In addition, $a_1 \leq 2 \theta_1$ so $\frac{a_1}{2} \leq \theta_1$, meaning that it is not involved in any of the first $(n - 1)/2$ after the first one. After the first $(n - 1)/2$ iterations of line 2 of \cref{algo:naive}, there is one remaining element in the range $(\theta_1, 2 \theta_1]$ because there are $n$ initially and each of the previous iterations removes 2 of them. By our definitions, this remaining element is $a_{(n + 1)/2}$. We claim that the density loss on the last while loop iteration is at most
$$\frac{1}{b_{(n + 1)/2}} - \frac{1}{a_{(n + 1)/2}}$$

If $a_{(n + 1)/2} < 2 b_{(n + 1)/2}$ then line 7 of \cref{algo:naive} is executed and the density loss above holds with equality since $b_{(n + 1)/2}$ is effectively replaced by $a_{(n + 1)/2}$. If $a_{(n + 1)/2} \geq 2b_{(n + 1)/2}$ then line 7 of \cref{algo:naive} is executed during line 11 of \cref{algo:imp} and we have a density loss of $\frac{1}{a_{(n + 1)/2}}$. However, since $b_{(n + 1)/2} \leq \frac{a_{(n + 1)/2}}{2}$ in this case,
$$\frac{1}{b_{(n + 1)/2}} - \frac{1}{a_{(n + 1)/2}} \geq \frac{2}{a_{(n + 1)/2}} - \frac{1}{a_{(n + 1)/2}} = \frac{1}{a_{(n + 1)/2}}$$

so in either case the density loss is upper bounded by 
$$\frac{1}{b_{(n + 1)/2}} - \frac{1}{a_{(n + 1)/2}}$$

Now, the overall density loss in this subcase is at most
$$\left( \frac{1}{b_{(n - 1)/2}} - \frac{1}{a_1} \right) + \left( \frac{1}{b_{(n + 1)/2}} - \frac{1}{a_{(n + 1)/2}} \right)$$

where we have added the density losses from all but the last iteration and the last iteration. As we have discussed, before the last iteration $A$ contained a $\frac{a_1}{2}$ so 
$$\frac{1}{b_{(n + 1)/2}} \leq \frac{2}{a_1}$$

In addition, because $b_{(n - 1)/2}$ is the second-largest element of $A$ on the $((n - 1)/2)$-th iteration and both of the largest elements are removed, on the next iteration, the $((n + 1)/2)$-th, the largest element, $a_{(n + 1)/2}$ must be at most $b_{(n - 1)/2}$. So,
$$\frac{1}{b_{(n - 1)/2}} \leq \frac{1}{a_{(n + 1)/2}}$$

Plugging these inequalities into the density loss formula, we have
$$\left( \frac{1}{b_{(n - 1)/2}} - \frac{1}{a_1} \right) + \left( \frac{1}{b_{(n + 1)/2}} - \frac{1}{a_{(n + 1)/2}} \right) \leq \frac{2}{a_1} - \frac{1}{a_1} + \frac{1}{a_{(n + 1)/2}} - \frac{1}{a_{(n + 1)/2}} = \frac{1}{a_1}$$

Finally, $a_1$ is the maximum element of $A$ before the while loop execution, and $m_3$ is an element of $A$, so $a_1 \geq m_3$. In addition, in this case $m_3 > \frac{4}{3} \theta_i$ so
$$\frac{1}{a_1} \leq \frac{1}{m_3} < \frac{3}{4 \theta_i}$$

This proves the desired density bound for this subcase, and we have now completed both subcases, so Case 2 is complete. Since $n$ is guaranteed to be either even or odd, this completes the proof.
\end{proof}

\subsection{\texorpdfstring{Proof of \cref{claim13}.}{Proof of Lemma 2.6.}}\label{proof:claim13}

\begin{proof}
Note that iterations of line 4 of \cref{algo:imp} continue until $\theta_{\text{current}}$ becomes $\theta$ (inclusive), and $\theta_{\text{current}}$ is halved at each step so there are $\log_2(\theta_1 / \theta) + 1$ iterations. We will prove the claim by induction on the number of completed iterations. In particular, we define
$$P(i) := D(A_i) \geq D(A) - \sum_{m = 1}^{i} \frac{3}{4 \theta_m}$$

For the base case $P(1)$, we have that $N_{\theta_1}(A) \geq 2$ so we can immediately apply \cref{claim12} to get that
$$D(A_1) \geq D(A) - \frac{3}{4 \theta_1}$$

$P(1)$ is the statement that
$$D(A_1) \geq D(A) - \sum_{m = 1}^{1} \frac{3}{4 \theta_m} = D(A) - \frac{3}{4 \theta_1}$$

so the claim result shows that $P(1)$ is true.Now, we want to prove that for any $t$ such that $1 \leq t < 1 + \log_2(\theta_1/\theta)$, $P(t)$ implies $P(t + 1)$. In other words, we are given that
$$D(A_t) \geq D(A) - \sum_{m = 1}^{t} \frac{3}{4 \theta_m}$$

We do casework depending on the value of $N_{\theta_{t + 1}}(A_t)$, aiming to prove that the density loss on the $(t + 1)$-th iteration is at most $3/(4 \theta_{t + 1})$ in any case.

\textbf{Case 1}: $N_{\theta_{t + 1}}(A_t) = 0$

In this case, there is no density loss because $n < 3$, so the condition of line 6 is not asserted, and the while loop condition within line 11 is false, so we exit. $0 \leq \frac{3}{4 \theta_{t + 1}}$, as desired.

\vspace{10pt}

\textbf{Case 2}: $N_{\theta_{t + 1}}(A_t) \geq 2$

We can apply \cref{claim12} to immediately get that
$$D(A_{t + 1}) \geq D(A_t) - \frac{3}{4 \theta_{t + 1}}$$

implying the desired density loss bound. 

\vspace{10pt}

\textbf{Case 3}: $N_{\theta_{t + 1}}(A_t) = 1$

We claim that in this case there must have been a thirding operation (i.e., line 10) on the previous iteration of the algorithm. Suppose, for the sake of contradiction, that there was no thirding operation in the previous iteration. Then, we claim that the smallest element generated (i.e., introduced by either line 10 of \cref{algo:imp} or line 7 of \cref{algo:imp}) by the algorithm across all previous iterations is greater than $\theta_{t + 1}$. This is because on an iteration where $\theta_{\text{current}} = \theta'$, line 10 generates values greater than $\frac{\theta'}{3}$ and line 7 of \cref{algo:naive} generates values greater than $\frac{\theta'}{2}$.

Since $\theta_{\text{current}}$ halves at each step, $\frac{\theta'}{2} \geq \theta_{t + 1}$ where $\theta'$ corresponds to any previous iteration. In addition, because no thirding occurs on the immediate previous iteration, the $\theta'$ values for which thirding occurs satisfy $\frac{\theta'}{4} \geq \theta_{t + 1}$. $\frac{\theta'}{3} > \frac{\theta'}{4}$ so this completes the proof that every element generated by the algorithm before iteration $t + 1$ is greater than $\theta_{t + 1}$. Now, since $p = 2 \theta_1$ was originally the minimum power of 2 such that $D(A \leq p) > D(W \leq p)$ and no elements were introduced in the range $(0, \theta_{t + 1}]$, $p = \theta_{t + 1}$ does not satisfy
$$D(A_t \leq p) > D(W \leq p)$$

In other words,
$$D(A_t \leq \theta_{t + 1}) \leq \sum_{i = 0}^{\log_2(\theta_{t}) - 2} \frac{1}{2^i + 1}$$

$N_{\theta_{t + 1}}(A_t) = 1$ implies there is a single element in the range $(\theta_{t + 1}, 2 \theta_{t + 1}]$, call it $d$. Recalling that $A_t$ is composed entirely of values at most $2 \theta_{t + 1}$, we have that
$$D(A_t) = D(A_t \leq \theta_{t + 1}) + \frac{1}{d}$$

Since $d > \theta_{t + 1}$, we have
$$D(A_t) \leq \frac{1}{\theta_{t + 1}} + \sum_{i = 0}^{\log_2(\theta_{t}) - 2} \frac{1}{2^i + 1}$$

By our inductive hypothesis, we have that
$$D(A_t) \geq D(A) - \sum_{m = 1}^{t} \frac{3}{4 \theta_m}$$

Putting the two inequalities above together, we have 
$$D(A) - \sum_{m = 1}^{t} \frac{3}{4 \theta_m} \leq \frac{1}{\theta_{t + 1}} + \sum_{i = 0}^{\log_2(\theta_{t}) - 2} \frac{1}{2^i + 1}$$

Recall that $p = 2 \theta_1$ satisfies $D(A \leq p) > D(W \leq p)$, so
$$D(A) > \sum_{i = 0}^{\log_2(\theta_1)} \frac{1}{2^i + 1}$$

Plugging this into the previous inequality, we have
$$\sum_{i = 0}^{\log_2(\theta_1)} \frac{1}{2^i + 1} - \sum_{i = 0}^{\log_2(\theta_{t}) - 2} \frac{1}{2^i + 1} \leq \frac{1}{\theta_{t + 1}} + \sum_{m = 1}^{t} \frac{3}{4 \theta_m}$$

Simplifying the left hand side, we have
$$\sum_{i = 0}^{\log_2(\theta_1)} \frac{1}{2^i + 1} - \sum_{i = 0}^{\log_2(\theta_{t}) - 2} \frac{1}{2^i + 1} = \sum_{i = \log_2(\theta_t) - 1}^{\log_2(\theta_1)} \frac{1}{2^i + 1}$$

Plugging this into the inequality, it becomes
$$\sum_{i = \log_2(\theta_t) - 1}^{\log_2(\theta_1)} \frac{1}{2^i + 1} - \sum_{m = 1}^{t} \frac{3}{4 \theta_m} \leq \frac{1}{\theta_{t + 1}}$$

Manipulating the left hand side, we have
$$\sum_{i = \log_2(\theta_t) - 1}^{\log_2(\theta_1)} \frac{1}{2^i + 1} - \sum_{m = 1}^{t} \frac{3}{4 \theta_m} = \frac{1}{\theta_t/2 + 1} + \sum_{i = \log_2(\theta_t)}^{\log_2(\theta_1)} \frac{1}{2^i + 1} - \sum_{m = 1}^{t} \frac{3}{4 \theta_m}$$
$$= \frac{1}{\theta_{t + 1} + 1} + \sum_{i = \log_2(\theta_t)}^{\log_2(\theta_1)} \left( \frac{1}{2^i + 1} - \frac{3}{4} \cdot \frac{1}{2^i} \right)$$

Plugging back into the inequality, it becomes
$$\frac{1}{\theta_{t + 1} + 1} + \sum_{i = \log_2(\theta_t)}^{\log_2(\theta_1)} \left( \frac{1}{2^i + 1} - \frac{3}{4} \cdot \frac{1}{2^i} \right) \leq \frac{1}{\theta_{t + 1}}$$

Note that we have assumed $\theta \geq 16$, so $\theta_t \geq 32$. So, the $i$ in the summation is at least five at every summand. Then, 
$$\frac{1}{2^i + 1} - \frac{3}{4} \cdot \frac{1}{2^i} = \frac{1}{4} \cdot \frac{1}{2^i} + \frac{1}{2^i + 1} - \frac{1}{2^i} = \frac{1}{4} \cdot \frac{1}{2^i} - \frac{1}{2^i + 1} \cdot \frac{1}{2^i} = \frac{1}{2^i} \cdot \left( \frac{1}{4} - \frac{1}{2^i + 1} \right) > 0$$

where we have used $i \geq 2$ in the last step. This means that
$$\sum_{i = \log_2(\theta_t)}^{\log_2(\theta_1)} \left( \frac{1}{2^i + 1} - \frac{3}{4} \cdot \frac{1}{2^i} \right) \leq \frac{1}{2^{\log_2(\theta_t)} + 1} - \frac{3}{4} \cdot \frac{1}{2^{\log_2(\theta_t)}} = \frac{1}{\theta_t + 1} - \frac{3}{4} \cdot \frac{1}{\theta_t}$$

Plugging this into the previous inequality, we have
$$\frac{1}{\theta_{t + 1} + 1} + \frac{1}{\theta_t + 1} - \frac{3}{4} \cdot \frac{1}{\theta_t} \leq \frac{1}{\theta_{t + 1}}$$

Recalling that $\theta_t = 2 \theta_{t + 1}$, the inequality becomes
$$\frac{1}{\theta_{t + 1} + 1} + \frac{1}{2 \theta_{t + 1} + 1} - \frac{3}{4} \cdot \frac{1}{2\theta_{t + 1}} \leq \frac{1}{\theta_{t + 1}}$$

Manipulating the inequality, it becomes
$$\left( \frac{1}{\theta_{t + 1}} - \frac{1}{\theta_{t + 1} + 1} \right) + \left( \frac{1}{2 \theta_{t + 1}} - \frac{1}{2 \theta_{t + 1} + 1} \right) \geq \frac{1}{4} \cdot \frac{1}{2 \theta_{t + 1}}$$

Manipulating further, we have
$$\frac{1}{\theta_{t + 1}  \cdot (\theta_{t + 1} + 1)} + \frac{1}{2 \theta_{t + 1} \cdot (2 \theta_{t + 1} + 1)} \geq \frac{1}{8 \theta_{t + 1}}$$

Recall that we have assumed $\theta \geq 16$, so we also have $\theta_{t + 1} \geq 16$. Since
$$\frac{1}{\theta_{t + 1}  \cdot (\theta_{t + 1} + 1)} + \frac{1}{2 \theta_{t + 1} \cdot (2 \theta_{t + 1} + 1)} \leq \frac{1}{\theta_{t + 1}^2} + \frac{1}{4 \theta_{t + 1}^2} = \frac{5}{4\theta_{t + 1}^2}$$

we can plug this into the inequality to get
$$\frac{5}{4 \theta_{t + 1}^2} \geq \frac{1}{8 \theta_{t + 1}}$$

Simplifying, we have
$$\theta_{t + 1} \leq \frac{40}{4} = 10$$

However, we have assumed that $\theta_{t + 1} \geq 16$, generating a contradiction. Thus, we have proven that in this case (where $N_{\theta_{t + 1}}(A_t) = 1$), a thirding must have occurred on the previous iteration. We now do casework based on whether $d$ (the lone element in $(\theta_{t + 1}, 2 \theta_{t + 1}]$) is greater or less than $\frac{4}{3} \theta_{t + 1}$, noting that there is only 1 iteration within line 11.

\vspace{10pt}

\textbf{Subcase 1}: $d \geq \frac{4}{3} \theta_{t + 1}$

In this case, if line 9 of \cref{algo:naive} is executed, then $d$ is deleted, resulting in a density loss of $\frac{1}{d}$. If line 7 of \cref{algo:naive} is executed, it must be that $b > \frac{a}{2}$, and so the density loss is
$$\frac{1}{b_1} - \frac{1}{a_1} \leq \frac{2}{a_1} = \frac{1}{a_1} = \frac{1}{d}$$

The density loss is at most $\frac{1}{d}$ in either situation and since $d \geq \frac{4}{3} \theta_{t + 1}$, the density loss in this subcase is at most $\frac{3}{4 \theta_{t + 1}}$, as desired. 

\vspace{10pt}

\textbf{Subcase 2}: $d < \frac{4}{3} \theta_{t + 1}$

We have proven that there was a thirding at the previous iteration, and the value that was divided by three was at least $2 \theta_{t + 1}$, meaning the thirded value is at least $\frac{2}{3} \theta_{t + 1}$. This means that $b_1 \geq \frac{2}{3} \theta_{t + 1}$. Since $d < \frac{4}{3} \theta_{t + 1} = 2 \cdot \frac{4}{3} \theta_{t + 1}$, line 7 of \cref{algo:naive} must be executed, for a corresponding density loss of
$$\frac{1}{b_1} - \frac{1}{a_1} \leq \frac{1}{b_1} - \frac{1}{2b_1} = \frac{1}{2b_1}$$

Recalling that $b_1 \geq \frac{2}{3} \theta_{t + 1}$, the density loss is at most $3/(4 \theta_{t + 1})$, as desired. Since this exhausts the set of possibilities for $d$, this completes the analysis for this case. In addition, since $N_{\theta_{t + 1}} = 0$, $N_{\theta_{t + 1}} = 1$ or $N_{\theta_{t + 1}} \geq 2$, we have considered all 3 possible cases. This means that the density loss on the $(t + 1)$-th iteration is at most $3/(4 \theta_{t + 1})$ in any case. So,
$$D(A_{t + 1}) \geq D(A_t) - \frac{3}{4 \theta_{t + 1}}$$

and from the inductive hypothesis we have that
$$D(A_t) \geq D(A) - \sum_{m = 1}^t \frac{3}{4 \theta_m}$$

so combining these inequalities we have
$$D(A_{t + 1}) \geq D(A) - \frac{3}{4 \theta_{t + 1}} - \sum_{m = 1}^t \frac{3}{4 \theta_m} = D(A) - \sum_{m = 1}^{t + 1} \frac{3}{4 \theta_m}$$

This is exactly $P(t + 1)$, so this completes the inductive step. By the principle of mathematical induction, $P(x)$ is true for all $x$ up to and including the last iteration, $x = 1 + \log_2(\theta_1 / \theta)$. $P(1 + \log_2(\theta_1 / \theta))$ is the statement that
$$D(A_{1 + \log_2(\theta_1 / \theta)}) \geq D(A) - \sum_{m = 1}^{1 + \log_2(\theta_1 / \theta)} \frac{3}{4 \theta_m}$$

Since this is the last iteration, by the definition of $A_i$, we have that
$$A_{1 + \log_2(\theta_1 / \theta)} = \text{cfold}^{\text{imp}}_{\theta}(A)$$

meaning our inequality becomes
$$D(\text{cfold}^{\text{imp}}_{\theta}(A)) \geq D(A) - \sum_{m = 1}^{1 + \log_2(\theta_1 / \theta)} \frac{3}{4 \theta_m}$$

This is exactly the statement of our claim, completing the proof.
\end{proof}

\subsection{\texorpdfstring{Proof of \cref{claim14}.}{Proof of Theorem 4.7.}}\label{proof:claim14}

\begin{proof}
We proceed with casework based on the value of $E$, which is guaranteed to exist due to \cref{claim9}.

\textbf{Case 1}: $E = 1$

$D(W \leq 1) = 0$, so $D(L \leq 1) > 0$. In this case, $L$ must contain at least one element that is a 1, and the instance $[1]$ is schedulable (it can be repeated infinitely), so by monotonicity, $L$ is schedulable as well.

\vspace{10pt}

\textbf{Case 2}: $E = 2$

$D(W \leq 2) = \frac{1}{2}$, so $D(L \leq 2) > \frac{1}{2}$. This means $L$ must contain either a 1, in which case it is schedulable, as before, or two 2's, also implying schedulability.

\vspace{10pt}

\textbf{Case 3}: $E = 4$

$D(W \leq 4) = \frac{1}{2} + \frac{1}{3} = \frac{5}{6}$, so $D(L \leq 4) > \frac{5}{6}$. 

We claim that $L' = \text{cfold}_{32}(L)$ is schedulable (note that this is not the improved fold operation). By \cref{claim4}, we have that 
$$D(L') \geq D(L) - \frac{2}{32}$$

We do casework depending on the value of $N_{16}(L')$ and whether a thirding operation occurs (i.e., line 10). 

\vspace{10pt}

\textbf{Subcase 1}: $N_{16}(L') \leq 1$

Note that the elements that are introduced by $\text{cfold}_{\theta}(A)$ are greater than $\theta/2$. In particular, in our case, it means that the elements of $D(L' \leq 16)$ were all present in $L$ to begin with. Now, we claim that $L'' = L' \leq 16$ is schedulable. In particular, note that 
$$D(L'') \geq D(L') - \frac{1}{16} \geq D(L) - \frac{2}{32} - \frac{1}{16} = D(L) - \frac{2}{16}$$

Recalling that $D(L) \geq \sum_{i = 0}^{\infty} \frac{1}{2^i + 1}$, by assumption, we have 
$$D(L'') \geq -\frac{2}{16} + \sum_{i = 0}^{\infty} \frac{1}{2^i + 1}$$

Since no elements were introduced in $L'$ that were less than or equal to 16, $L''$ is entirely composed of integers, and $D(L'' \leq 4) = D(L \leq 4) > \frac{5}{6}$. Note that there are only a finite number of integer pinwheel instances $A$ satisfying the following conditions, where we denote the maximum element of $A$ as $m$:

\begin{itemize}
    \item $m \leq 16$
    \item $D(A \leq 4) > \frac{5}{6}$
    \item $D(A) \geq -\frac{2}{16} + \sum_{i = 0}^{9} \frac{1}{2^i + 1}$
    \item $D(A \ominus m) < -\frac{2}{16} + \sum_{i = 0}^{9} \frac{1}{2^i + 1}$
\end{itemize}

The last condition prevents us from considering infinitely long instances. We have enumerated and found a solution to all possible $A$ satisfying the four conditions above. This means that $L''$ is schedulable because we already know that it satisfies the first three conditions, so we can remove the maximum element until it satisfies the fourth condition, resulting in schedulability, and by the monotonicity property, this means $L''$ is schedulable as well. Again applying the monotonicity property, this means $L'$ is schedulable, as desired. 

\vspace{10pt}

\textbf{Subcase 2}: $N_{16}(L') \geq 2$ and $\text{cfold}_{16}^{\text{imp}}(L')$ does not involve a thirding operation

In this case, we claim that $L'' = \text{cfold}_{16}^{\text{imp}}(L')$ is schedulable. Note that by \cref{claim12}, we have
$$D(L'') \geq D(L') - \frac{3}{4} \cdot \frac{1}{16}$$

Putting this together with $D(L') \geq D(L) - \frac{1}{16}$ we have 
$$D(L'') \geq D(L) - \frac{7}{4} \cdot \frac{1}{16} \geq -\frac{7}{4} \cdot \frac{1}{16} + \sum_{i = 0}^{\infty} \frac{1}{2^i + 1}$$

Recall our $D'(A)$ definition from \cref{claim5} as
\[
D'(A) = \sum_{i = 1}^n
\begin{cases} 
1/a_i & \text{if } a_i \leq 8, \\
1/(a_i - 1) & \text{if } a_1 > 8.
\end{cases}
\]

Note that there are a finite number of integer instances $A$ with maximum $m$ satisfying

\begin{itemize}
    \item $m \leq 16$
    \item $D(A \leq 4) > \frac{5}{6}$
    \item $D'(A) \geq -\frac{7}{4} \cdot \frac{1}{16} + \sum_{i = 0}^{9} \frac{1}{2^i + 1}$
    \item $D'(A \ominus m) < -\frac{7}{4} \cdot \frac{1}{16} + \sum_{i = 0}^{9} \frac{1}{2^i + 1}$
\end{itemize}

We have enumerated and found a solution to all possible $A$ satisfying the four conditions above. Then, $L'''$, which is created by taking the ceiling of every element of $L''$ and removing the maximum until it satisfies the last condition, is schedulable. Therefore, by the monotonicity property, $L''$ is schedulable as well. Note that the $D'$ function accounts for the density change due to taking ceilings, and we only need to do $1/(a_i - 1)$ for $a > 8$ because if no thirding occurs then every element generated by $\text{cfold}_{16}^{\text{imp}}(L')$ is greater than 8. Since $L''$ is schedulable, by \cref{claim11}, $L'$ is schedulable, as desired. 

\vspace{10pt}

\textbf{Subcase 3}: $N_{16}(L') \geq 2$ and $\text{cfold}_{16}^{\text{imp}}(L')$ involves a thirding operation

In this case, we claim that a slight modification of $\text{cfold}_{16}^{\text{imp}}(L')$ is schedulable. In particular, since a thirding operation occurred and the while loop in line 4 iterates once, with at most one thirding operation per iteration, there is a unique element $v$ generated by the thirding. 

We claim that $L'' = (\text{cfold}_{16}^{\text{imp}}(L') \ominus v) \sqcup (3v, 3v, 3v)$ is schedulable. Note that we have essentially reversed the thirding operation, and if $L''$ is schedulable then $L'$ is schedulable because the combined execution of lines 9 and 10 and our reversal replaces $(m_1, m_2)$ with $(m_3, m_3)$, preserving unschedulability, and the rest of our operations also preserve unschedulability, as discussed in \cref{claim11}. In addition, applying \cref{claim12} we have
$$D(L'') = D(\text{cfold}_{16}^{\text{imp}}(L')) \geq D(L') - \frac{3}{4} \cdot \frac{1}{16} \geq D(A) - \frac{1}{16} - \frac{3}{4} \cdot \frac{1}{16} \geq -\frac{7}{4} \cdot \frac{1}{16} + \sum_{i = 0}^{\infty} \frac{1}{2^i + 1}$$

Note that there are a finite number of integer instances $A$ with top 4 maximum elements $m_1 \geq m_2 \geq m_3 \geq m_4$ satisfying

\begin{itemize}
    \item $m_1 = m_2 = m_3$
    \item $17 \leq m_1 \leq 22$
    \item $m_4 \leq 16$
    \item $D(A \leq 4) > \frac{5}{6}$
    \item $D'(A) \geq -\frac{7}{4} \cdot \frac{1}{16} + \sum_{i = 0}^{9} \frac{1}{2^i + 1}$
    \item $D'(A \ominus m_4) < -\frac{7}{4} \cdot \frac{1}{16} + \sum_{i = 0}^{9} \frac{1}{2^i + 1}$
\end{itemize}

We have enumerated and found a solution to all possible $A$ satisfying the four conditions above. 

Then, $L'''$, which is created by taking the ceiling of every element of $L''$ and removing the $m_4$ element (maximum under 17) until it satisfies the last condition, is schedulable, so by the monotonicity property, $L''$ is schedulable as well. As discussed previously, this implies $L'$ is schedulable, as desired. 

The three subcases above exhaust all possibilities, so this proves that $L'$ is schedulable in any subcase. By the contrapositive of \cref{claim3}, this means $L$ is schedulable, as desired.

\vspace{10pt}

\textbf{Case 4}: $E = 8$

$D(W \leq 8) = \frac{1}{2} + \frac{1}{3} + \frac{1}{5} = \frac{31}{30}$, so $D(L \leq 8) > \frac{31}{30}$. We follow a similar procedure to the previous case, simply replacing the $D(A \leq 4) > \frac{5}{6}$ condition with $D(A \leq 8) > \frac{31}{30}$. We do casework depending on the value of $N_{16}(L')$ and whether a thirding operation occurs (i.e., line 10). 

\vspace{10pt}

\textbf{Subcase 1}: $N_{16}(L') \leq 1$

Note that the elements that are introduced by $\text{cfold}_{\theta}(A)$ are greater than $\theta/2$. In particular, in our case, it means that the elements of $D(L' \leq 16)$ were all present in $L$ to begin with. Now, we claim that $L'' = L' \leq 16$ is schedulable. 

In particular, note that 
$$D(L'') \geq D(L') - \frac{1}{16} \geq D(L) - \frac{2}{32} - \frac{1}{16} = D(L) - \frac{2}{16}$$

Recalling that $D(L) \geq \sum_{i = 0}^{\infty} \frac{1}{2^i + 1}$, by assumption, we have 
$$D(L'') \geq -\frac{2}{16} + \sum_{i = 0}^{\infty} \frac{1}{2^i + 1}$$

Since no elements were introduced in $L'$ that were less than or equal to 16, $L''$ is entirely composed of integers, and $D(L'' \leq 8) = D(L \leq 8) > \frac{31}{30}$. Note that there are only a finite number of integer pinwheel instances $A$ satisfying the following conditions, where we denote the maximum element of $A$ as $m$:

\begin{itemize}
    \item $m \leq 16$
    \item $D(A \leq 8) > \frac{31}{30}$
    \item $D(A) \geq -\frac{2}{16} + \sum_{i = 0}^{9} \frac{1}{2^i + 1}$
    \item $D(A \ominus m) < -\frac{2}{16} + \sum_{i = 0}^{9} \frac{1}{2^i + 1}$
\end{itemize}

We have enumerated and found a solution to all possible $A$ satisfying the four conditions above. This means that $L''$ is schedulable because we already know that it satisfies the first three conditions, so we can remove the maximum element until it satisfies the fourth condition, resulting in schedulability, and by the monotonicity property, this means $L''$ is schedulable as well. Again applying the monotonicity property, this means $L'$ is schedulable, as desired. 

\vspace{10pt}

\textbf{Subcase 2}: $N_{16}(L') \geq 2$ and $\text{cfold}_{16}^{\text{imp}}(L')$ does not involve a thirding operation

In this case, we claim that $L'' = \text{cfold}_{16}^{\text{imp}}(L')$ is schedulable. Note that by \cref{claim12}, we have
$$D(L'') \geq D(L') - \frac{3}{4} \cdot \frac{1}{16}$$

Putting this together with $D(L') \geq D(L) - \frac{1}{16}$ we have 
$$D(L'') \geq D(L) - \frac{7}{4} \cdot \frac{1}{16} \geq -\frac{7}{4} \cdot \frac{1}{16} + \sum_{i = 0}^{\infty} \frac{1}{2^i + 1}$$

Recall our $D'(A)$ definition from \cref{claim5} as
\[
D'(A) = \sum_{i = 1}^n
\begin{cases} 
1/a_i & \text{if } a_i \leq 8, \\
1/(a_i - 1) & \text{if } a_1 > 8.
\end{cases}
\]

Note that there are a finite number of integer instances $A$ with maximum $m$ satisfying

\begin{itemize}
    \item $m \leq 16$
    \item $D(A \leq 8) > \frac{31}{30}$
    \item $D'(A) \geq -\frac{7}{4} \cdot \frac{1}{16} + \sum_{i = 0}^{9} \frac{1}{2^i + 1}$
    \item $D'(A \ominus m) < -\frac{7}{4} \cdot \frac{1}{16} + \sum_{i = 0}^{9} \frac{1}{2^i + 1}$
\end{itemize}

We have enumerated and found a solution to all possible $A$ satisfying the four conditions above. 

Then, $L'''$, which is created by taking the ceiling of every element of $L''$ and removing the maximum until it satisfies the last condition, is schedulable, so by the monotonicity property, $L''$ is schedulable as well. 

Note that the $D'$ function accounts for the density change due to taking ceilings, and we only need to do $1/(a_i - 1)$ for $a > 8$ because if no thirding occurs then every element generated by $\text{cfold}_{16}^{\text{imp}}(L')$ is greater than 8.

Since $L''$ is schedulable, by \cref{claim11}, $L'$ is schedulable, as desired. 

\vspace{10pt}

\textbf{Subcase 3}: $N_{16}(L') \geq 2$ and $\text{cfold}_{16}^{\text{imp}}(L')$ involves a thirding operation

In this case, we claim that a slight modification of $\text{cfold}_{16}^{\text{imp}}(L')$ is schedulable. In particular, since a thirding operation occurred and the while loop in line 5 iterates once, with at most one thirding operation per iteration, there is a unique element $v$ generated by the thirding. 

We claim that $L'' = (\text{cfold}_{16}^{\text{imp}}(L') \ominus v) \sqcup (3v, 3v, 3v)$ is schedulable. Note that we have essentially reversed the thirding operation, and if $L''$ is schedulable, then $L'$ is schedulable because the combined execution of lines 9 and 10 and our reversal replaces $(m_1, m_2)$ with $(m_3, m_3)$, preserving unschedulability, and the rest of our operations also preserve unschedulability, as discussed in \cref{claim11}. In addition, by \cref{claim12} we have
$$D(L'') = D(\text{cfold}_{16}^{\text{imp}}(L')) \geq D(L') - \frac{3}{4} \cdot \frac{1}{16} \geq D(A) - \frac{1}{16} - \frac{3}{4} \cdot \frac{1}{16} \geq -\frac{7}{4} \cdot \frac{1}{64} + \sum_{i = 0}^{\infty} \frac{1}{2^i + 1}$$

Note that there are a finite number of integer instances $A$ with top 4 maximum elements $m_1 \geq m_2 \geq m_3 \geq m_4$ satisfying

\begin{itemize}
    \item $m_1 = m_2 = m_3$
    \item $17 \leq m_1 \leq 22$
    \item $m_4 \leq 16$
    \item $D(A \leq 8) > \frac{31}{30}$
    \item $D'(A) \geq -\frac{7}{4} \cdot \frac{1}{16} + \sum_{i = 0}^{9} \frac{1}{2^i + 1}$
    \item $D'(A \ominus m_4) < -\frac{7}{4} \cdot \frac{1}{16} + \sum_{i = 0}^{9} \frac{1}{2^i + 1}$
\end{itemize}

We have enumerated and found a solution to all possible $A$ satisfying the six conditions above. Then, $L'''$, which is created by taking the ceiling of every element of $L''$ and removing the $m_4$ element (maximum under 17) until it satisfies the last condition, is schedulable, so by the monotonicity property, $L''$ is schedulable as well. As discussed previously, this implies $L'$ is schedulable, as desired. 

The three subcases above exhaust all possibilities, so this proves that $L'$ is schedulable in any subcase. By the contrapositive of \cref{claim3}, this means $L$ is schedulable, as desired.

\vspace{10pt}

\textbf{Case 5}: $E = 16$

$D(W \leq 16) = \frac{1}{2} + \frac{1}{3} + \frac{1}{5} + \frac{1}{9} = \frac{103}{90}$, so $D(L \leq 16) > \frac{103}{90}$.

We claim that $L' = L \leq 16$ is schedulable. Note that every element of $L'$ is an integer and elements are no greater than 16.

Note that there are a finite number of integer instances $A$ with maximum $m$ satisfying

\begin{itemize}
    \item $m \leq 16$
    \item $D(A \leq 16) > \frac{103}{90}$
    \item $D(A \ominus m) \leq \frac{103}{90}$
\end{itemize}

We have enumerated and found a solution to all possible $A$ satisfying the 3 conditions above. This means that $L'$ is schedulable because we already know that it satisfies the first two conditions so we can remove the maximum element until it satisfies the third condition, resulting in schedulability, and by the monotonicity property this means $L'$ is schedulable as well. Again applying the monotonicity property, this means $L$ is schedulable, as desired. 

\vspace{10pt}

\textbf{Case 6}: $E = 32$

$D(W \leq 32) = \frac{1}{2} + \frac{1}{3} + \frac{1}{5} + \frac{1}{9} + \frac{1}{17} = \frac{1841}{1530}$, so $D(L \leq 32) > \frac{1841}{1530}$. We claim that $L' = L \leq 32$ is schedulable. To prove this, we will show that $L'' = \text{cfold}_{16}^{\text{imp}}(L')$ is schedulable. Note that by \cref{claim12},
$$D(L'') \geq D(L') - \frac{3}{4} \cdot \frac{1}{16} > \frac{1841}{1530} - \frac{3}{4} \cdot \frac{1}{16}$$

Now, we identify the appropriate modified density function for computer analysis in this case, associated with taking the ceiling of every element of $L''$. We first account for the thirding operation. 

Note that it is not possible for any generated elements to be less than 5. In the range $(5, 6)$, every generated element must be at least $\frac{17}{3}$ (since $m_1, m_2, m_3 > 16$, and we divide by 3). In the range $(6, 7)$, every generated element must be at least $\frac{19}{3}$, since we start with integers, so $m_3$ cannot be strictly between 18 and 19. No elements can be generated in the range $(7, 8)$ because $m_3$ cannot be strictly between 21 and 22, and the thirding operation only occurs if $m_3 \leq \frac{4}{3} \cdot 16 \approx 21.33$. In addition, for the halving operation, $L'$ is composed entirely of integers, so the elements of $L''$ must all be half-integer (if not an integer). In other words, for every $a$ such that $8 \leq a \leq 15$, the minimum possible element generated in the range $(a, a + 1)$ is $a + \frac{1}{2}$. Therefore, the appropriate density function is
\[
D_{\text{mod}}(A) = \sum_{i = 1}^n
\begin{cases} 
1/a_i & \text{if } a_i \leq 5, \\
3/17 & \text{ if } a_i = 6, \\
3/19 & \text{ if } a_i = 7, \\
1/8 & \text{ if } a_i = 8, \\
2/(2a_i - 1) & \text{if } a_1 > 8.
\end{cases}
\]

Note that there are only a finite number of integer pinwheel instances $A$ satisfying the following conditions, where we denote the maximum element of $A$ as $m$:

\begin{itemize}
    \item $m \leq 16$
    \item $D_{\text{mod}}(A) > \frac{1841}{1530} - \frac{3}{4} \cdot \frac{1}{16}$
    \item $D_{\text{mod}}(A \ominus m) \leq \frac{1841}{1530} - \frac{3}{4} \cdot \frac{1}{16}$
\end{itemize}

We have enumerated and found a solution to all possible $A$ satisfying the three conditions above. Then $L'''$, which is created by taking the ceiling of every element of $L''$ and removing the maximum until the last condition is satisfied, is schedulable. By \cref{claim11}, this means $L'$ is schedulable. Again applying the monotonicity property, this means $L$ is schedulable, as desired. 

\vspace{10pt}

\textbf{Case 7}: $E \geq 64$

By the structure of $W$, we have that
$$D(W \leq E) = \frac{1}{2} + \frac{1}{3} + \frac{1}{5} + \dots + \frac{1}{E/2 + 1} = \sum_{i = 0}^{\log_2(E) - 1} \frac{1}{2^i + 1}$$

We claim that $L' = L \leq E$ is schedulable. Note that 
$$D(L') = D(L \leq E) > \sum_{i = 0}^{\log_2(E) - 1} \frac{1}{2^i + 1}$$

We proceed with casework. Consider the execution of the last line of the while loop on line 5 of the algorithm for $\text{cfold}_{16}^{\text{imp}}(L')$. The value of $n$ calculated on the last iteration is either even or odd. In the case where it is odd, we further subdivide into whether a thirding operation occurs on the last step or not.

\vspace{10pt}

\textbf{Subcase 1}: $n$ is even

In this case, we claim that $L'' = \text{cfold}_{16}^{\text{imp}}(L')$ is schedulable. Note that by the iterative nature of the algorithm, 
$$\text{cfold}_{16}^{\text{imp}}(L') = \text{cfold}_{16}^{\text{imp}}((\text{cfold}_{32}^{\text{imp}}(L')))$$

We will use the second formulation to explore the density of $L''$. In particular, applying \cref{claim12}, and recalling that $n$ is even, we have
$$D(L'') = D(\text{cfold}_{16}^{\text{imp}}((\text{cfold}_{32}^{\text{imp}}(L')))) \geq D(\text{cfold}_{32}^{\text{imp}}(L')) - \frac{1}{2} \cdot \frac{1}{16}$$

Applying \cref{claim13}, we have
$$D(\text{cfold}_{32}^{\text{imp}}(L')) \geq D(L') - \sum_{m = 1}^{1 + \log_2(\theta_1 / 32)} \frac{3}{4 \theta_m}$$

Since the maximum of $L'$ is $E$, and $N_{E/2}(L') > 0$, $\theta_1 = E/2$. Plugging this in, we have
$$D(L') - \sum_{m = 1}^{1 + \log_2(\theta_1 / 32)} \frac{3}{4 \theta_m} > \sum_{i = 0}^{\log_2(E) - 1} \frac{1}{2^i + 1} - \sum_{m = 1}^{\log_2(E) - 5} \frac{3}{4 \theta_m}$$

Putting everything together, we have
$$D(L'') > -\frac{1}{2} \cdot \frac{1}{16} + \sum_{i = 0}^{\log_2(E) - 1} \frac{1}{2^i + 1} - \sum_{m = 1}^{\log_2(E) - 5} \frac{3}{4 \theta_m}$$
$$= -\frac{1}{32} + \frac{1}{2} + \frac{1}{3} + \frac{1}{5} + \frac{1}{9} + \frac{1}{17} + \sum_{i = 5}^{\log_2(E) - 1} \frac{1}{2^i + 1} - \sum_{m = 1}^{\log_2(E) - 5} \frac{3}{4 \theta_m}$$

Seeing that the two sums have the sum number of terms and that the last value of $\theta$ (i.e., $\theta_{\log_2(E) - 5}$) for $\text{cfold}_{32}^{\text{imp}}(L')$ is 32, we have
$$-\frac{1}{32} + \frac{1}{2} + \frac{1}{3} + \frac{1}{5} + \frac{1}{9} + \frac{1}{17} + \sum_{i = 5}^{\log_2(E) - 1} \frac{1}{2^i + 1} - \sum_{m = 1}^{\log_2(E) - 5} \frac{3}{4 \theta_m} = \frac{28691}{24480} + \sum_{i = 5}^{\log_2(E) - 1} \left( \frac{1}{2^i + 1} - \frac{3}{4} \cdot \frac{1}{2^i} \right)$$

Since $E \geq 64$, $\log_2(E) - 1 \geq \log_2(64) - 1 = 6 - 1 = 5$. In addition, we have previously shown that if $i \geq 2$, the summand is positive:
$$\frac{1}{2^i + 1} - \frac{3}{4} \cdot \frac{1}{2^i} = \frac{1}{4} \cdot \frac{1}{2^i} + \frac{1}{2^i + 1} - \frac{1}{2^i} = \frac{1}{4} \cdot \frac{1}{2^i} - \frac{1}{2^i + 1} \cdot \frac{1}{2^i} = \frac{1}{2^i} \cdot \left( \frac{1}{4} - \frac{1}{2^i + 1} \right) > 0$$

Therefore, 
$$\sum_{i = 5}^{\log_2(E) - 1} \left( \frac{1}{2^i + 1} - \frac{3}{4} \cdot \frac{1}{2^i} \right) \geq \frac{1}{2^5 + 1} - \frac{3}{4} \cdot \frac{1}{2^5} = \frac{1}{33} - \frac{3}{4} \cdot \frac{1}{32} = \frac{29}{4224}$$

Putting everything together, we have
$$D(L'') > \frac{28691}{24480} + \frac{29}{4224} = \frac{1269799}{1077120}$$

Note that there are a finite number of integer instances $A$ with maximum denoted as $m$ satisfying

\begin{itemize}
    \item $m \leq 16$
    \item $D'(A) > \frac{1269799}{1077120}$
    \item $D'(A \ominus m) \leq \frac{1269799}{1077120}$
\end{itemize}

We have enumerated and found a solution to all possible $A$ satisfying the three conditions above. Then, $L'''$, which is created by taking the ceiling of every element of $L''$ and removing the maximum element until it satisfies the last condition, is schedulable, so by the monotonicity property, $L''$ is schedulable as well. 

By \cref{claim11}, this means $L'$ is schedulable as well, as desired.

\vspace{10pt}

\textbf{Subcase 2}: $n$ is odd and no thirding operation occurs on the last iteration

In this case, we claim that $L'' = \text{cfold}_{16}^{\text{imp}}(L')$ is schedulable. Applying \cref{claim13}, we have
$$D(\text{cfold}_{16}^{\text{imp}}(L')) \geq D(L') - \sum_{m = 1}^{1 + \log_2(\theta_1 / 16)} \frac{3}{4 \theta_m}$$

Since the maximum of $L'$ is $E$, and $N_{E/2}(L') > 0$, $\theta_1 = E/2$. Plugging this in, we have
$$D(L') - \sum_{m = 1}^{1 + \log_2(\theta_1 / 16)} \frac{3}{4 \theta_m} > \sum_{i = 0}^{\log_2(E) - 1} \frac{1}{2^i + 1} - \sum_{m = 1}^{\log_2(E) - 4} \frac{3}{4 \theta_m}$$
$$= \frac{1}{2} + \frac{1}{3} + \frac{1}{5} + \frac{1}{9} + \sum_{i = 4}^{\log_2(E) - 1} \frac{1}{2^i + 1} - \sum_{m = 1}^{\log_2(E) - 4} \frac{3}{4 \theta_m}$$

As we have done in subcase 1, we know that
$$\sum_{i = 4}^{\log_2(E) - 1} \frac{1}{2^i + 1} - \sum_{m = 1}^{\log_2(E) - 4} \frac{3}{4 \theta_m} = \sum_{i = 4}^{\log_2(E) - 1} \left( \frac{1}{2^i + 1} - \frac{3}{4} \cdot \frac{1}{2^i} \right)$$

Since $E \geq 64$, $\log_2(E) - 1 \geq \log_2(64) - 1 = 6 - 1 = 5$. In addition, we have previously shown that if $i \geq 2$, the summand is positive. Therefore, 
$$\sum_{i = 4}^{\log_2(E) - 1} \left( \frac{1}{2^i + 1} - \frac{3}{4} \cdot \frac{1}{2^i} \right) \geq \sum_{i = 4}^{5} \left( \frac{1}{2^i + 1} - \frac{3}{4} \cdot \frac{1}{2^i} \right) = \left( \frac{1}{17} - \frac{3}{4} \cdot \frac{1}{16} \right) + \left( \frac{1}{33} - \frac{3}{4} \cdot \frac{1}{32} \right)$$

Putting everything together, we have
$$D(L'') > \frac{1}{2} + \frac{1}{3} + \frac{1}{5} + \frac{1}{9} + \left( \frac{1}{17} - \frac{3}{4} \cdot \frac{1}{16} \right) + \left( \frac{1}{33} - \frac{3}{4} \cdot \frac{1}{32} \right) = \frac{1252969}{1077120}$$

Let us now define a generalized notion of the $D'$ function in the following way:
\[
D_{\text{c}}(A) = \sum_{i = 1}^n
\begin{cases} 
1/a_i & \text{if } a_i \leq c, \\
1/(a_i - 1) & \text{if } a_1 > c.
\end{cases}
\]

We will use this new function to make use of the following observation: there can be at most one element generated by $\text{cfold}_{16}^{\text{imp}}(L')$ that is less than 10. This is because no thirding operation occurs, so either $n < 3$ (where $n$ denotes the value of $n$ at the last iteration) in which case there can be only one generated element at all, or $m_3 \geq \frac{4}{3} \cdot 16 \approx 21.33$, in which case there are at most two elements less than 20, and elements are halved so there can be at most one element less than 10. We will now formulate the integer conditions to account for this improvement:

There are a finite number of integer instances $A$ with maximum denoted as $m$ satisfying

\begin{itemize}
    \item $m \leq 16$
    \item If $A$ contains a 9, $D_{10}(A) > \frac{1252969}{1077120} - \frac{1}{72}$ and $D_{10}(A \ominus m) \leq \frac{1252969}{1077120} - \frac{1}{72}$
    \item Else if $A$ contains a 10, $D_{10}(A) > \frac{1252969}{1077120} - \frac{1}{90}$ and $D_{10}(A \ominus m) \leq \frac{1252969}{1077120} - \frac{1}{90}$
    \item Else, $D_{10}(A) > \frac{1252969}{1077120}$ and $D_{10}(A \ominus m) \leq \frac{1252969}{1077120}$
\end{itemize}

We have enumerated and found a solution to all possible $A$ satisfying the three conditions above. Then, we claim that $L'''$, which is created by taking the ceiling of every element of $L''$ and removing the maximum element until it satisfies the last condition, is schedulable. Note that if $L'''$ contains a 9, the value in $L''$ it corresponds to could be arbitrarily close to 8 (but only one such 9 could, as discussed previously), so the density of $L''$ could exceed that of $D(L''')$ by as much as $\frac{1}{8} - \frac{1}{9} = \frac{1}{72}$. If $L'''$ instead contains a 10, the corresponding parameter is $\frac{1}{9} - \frac{1}{10} = \frac{1}{90}$. If $L'''$ contains both, the $\frac{1}{72}$ parameter is still appropriate since at most one 9 or 10 can be generated (i.e., subject to modified density). Now, by the monotonicity property, $L''$ is schedulable. By \cref{claim11}, this means $L'$ is schedulable as well, as desired. 

\vspace{10pt}

\textbf{Subcase 3}: 

As in Case 4, Subcase 3, we claim that a slight modification of $\text{cfold}_{16}^{\text{imp}}(L')$ is schedulable. In particular, since a thirding operation occurred and the while loop in line 5 iterates once, with at most one thirding operation per iteration, there is a unique element $v$ generated by the thirding. We claim that $L'' = (\text{cfold}_{16}^{\text{imp}}(L') \ominus v) \sqcup (3v, 3v, 3v)$ is schedulable. 

Note that we have essentially reversed the thirding operation, and if $L''$ is schedulable, then $L'$ is schedulable because the combined execution of lines 9 and 10 and our reversal replaces $(m_1, m_2)$ with $(m_3, m_3)$, preserving unschedulability, and the rest of our operations also preserve unschedulability, as discussed in \cref{claim11}. We can apply all the same steps as in Subcase 2 to derive
$$D(\text{cfold}_{16}^{\text{imp}}(A)) > \frac{1252969}{1077120}$$

In addition, $\frac{1}{v} = \frac{1}{3v} + \frac{1}{3v} + \frac{1}{3v}$ so
$$D(L'') = D((\text{cfold}_{16}^{\text{imp}}(L') \ominus v) \sqcup (3v, 3v, 3v)) = D(\text{cfold}_{16}^{\text{imp}}(L')) > \frac{1252969}{1077120}$$

Furthermore, note that because $3v$ is the third smallest element in the range $(16, 32]$ are thirded (not halved), every generated element is greater than $3v/2$. Note that there are a finite number of integer instances $A$ with top 4 maximum elements $m_1 \geq m_2 \geq m_3 \geq m_4$ satisfying

\begin{itemize}
    \item $m_1 = m_2 = m_3$
    \item $17 \leq m_1 \leq 22$
    \item $m_4 \leq 16$
    \item If $m_1 = 17$ or $m_1 = 18$, $D_{8}(A) > \frac{1252969}{1077120}$ and $D_{8}(A \ominus m_4) \leq \frac{1252969}{1077120}$
    \item If $m_1 = 19$ or $m_1 = 20$, $D_{9}(A) > \frac{1252969}{1077120}$ and $D_{9}(A \ominus m_4) \leq \frac{1252969}{1077120}$
    \item If $m_1 = 21$ or $m_1 = 22$, $D_{10}(A) > \frac{1252969}{1077120}$ and $D_{10}(A \ominus m_4) \leq \frac{1252969}{1077120}$
\end{itemize}

We have enumerated and found a solution to all possible $A$ satisfying the conditions above. Then, we claim that $L'''$, which is created by taking the ceiling of every element of $L''$ and removing the maximum element (less than 17) until it satisfies the last condition, is schedulable. $D_8(A) = D'(A)$ is the default modified density function. We can use $D_9$ if $m_1$ is 19 or 20 because the value that we take the ceiling of must be greater than 18, so the generated elements must be greater than $3v/2 \geq 9$ similarly, we can use $D_{10}$ if $m_1$ is 21 or 22 because it corresponds to the rest of the elements of $L''$ being greater than 20. By the monotonicity property, $L''$ is schedulable. Now, by \cref{claim11}, we have that $L'$ is schedulable as desired. 

These subcases exhaust all possibilities, so we have proven that $L'$ is schedulable in any case. By the monotonicity property, this means that $L$ is schedulable as well, as desired.

This completes the analysis of our 7 cases, and $E$ is guaranteed to exist, so we have exhausted all possible cases, completing the proof. 
\end{proof}

\section{\texorpdfstring{Omitted Proofs in Section \ref{sec:bgt}.}{Omitted Proofs in Section 5.}}

\subsection{\texorpdfstring{Proof of \cref{lemma:bgt1}.}{Proof of Lemma 5.1.}}\label{proof:bgt1}

\begin{proof}
Since $D(A) > 1 - \frac{1}{ab^2}$, by the Axiom of Archimedes, there exists $n \in \mathbb{N}$ such that $D(A) \geq 1 - \frac{1}{ab^2} + \frac{1}{n}$. Suppose, for the sake of contradiction, that there exists an instance $A$ satisfying the conditions of the lemma that is schedulable. Then, there is a valid scheduling of jobs, which we refer to below.

Suppose, for the sake of contradiction, there exists $k \in \mathbb{N}$ such that both jobs 1 and 2 are arithmetically scheduled on every day between $abk$ and $ab(k + 1) - 1$ inclusive. Then, there exists $c, d \in \{ 0, 1, \dots, a - 1 \} $ such that job $a$ is scheduled on days $abk + c$, $abk + c + a$, $abk + c + 2a$, $\dots$ and job $b$ is scheduled on days $abk + d$, $abk + d + b$, $abk + d + 2b$, $\dots$. Then, since $\gcd(a, b) = 1$, $abk \equiv 0 \pmod{ab}$, and $ab(k + 1) - 1 \equiv ab - 1 \pmod{ab}$, by the Chinese Remainder Theorem, there exists a unique $x \in \{ abk, \dots, ab(k + 1) - 1 \}$ such that $x \equiv c \pmod{a}$ and $x \equiv d \pmod{b}$. On day $x$, both jobs 1 and 2 are thus scheduled, but this violates the pinwheel constraint of one job per day, a contradiction. 

If job 1 is not scheduled arithmetically, then it must be that there are two consecutive schedulings with strictly less than a gap of $a$ between them, and similarly a gap less than $b$ for job 2. For such a gap strictly smaller than the period length, we refer to it as a ``left-push'', with appropriate multiplicity (e.g., a gap of $a - 2$ for job 1 counts as 2 left pushes). We have already proven that for every $k$, there is a left-push in the interval $[abk, ab(k + 1) - 1]$. 

Let $P = n \prod_{i = 1}^k a_i$, and consider the scheduling of jobs 1 and 2 between day 1 and day $P$. Partitioning these $P$ days into $\frac{P}{ab}$ groups of $ab$ contiguous days each, there are at least $\frac{P}{ab}$ left-pushes. There are a minimum of $\frac{P}{b}$ schedulings of job 1 in this interval, and for every $a$ left-pushes of job 1 there is an additional scheduling needed. Similarly, job 2 is scheduled a minimum of $\frac{P}{b}$ times with an additional scheduling per $b$ left-pushes of job 2. Our formula for the total number of schedulings of jobs 1 and 2 is then
$$\frac{P}{a} + \frac{P}{b} + \left\lfloor \frac{x}{a} \right\rfloor + \left\lfloor \frac{x}{b} \right\rfloor \geq \frac{P}{a} + \frac{P}{b} + \frac{x}{a} - 1 + \frac{y}{b} - 1 = \frac{P}{a} + \frac{P}{b} + \frac{bx + ay}{ab} - 2$$

where $x$ is the number of left-pushes for job 1 and $y$ is the number of left-pushes for job 2. Then, $x + y \geq \frac{P}{ab}$, so
$$\frac{P}{a} + \frac{P}{b} + \frac{bx + ay}{ab} - 2 \geq \frac{P}{a} + \frac{P}{b} + \frac{ax + ay}{ab} - 2 = \frac{P}{a} + \frac{P}{b} + \frac{P}{ab^2} - 2$$

Now, since $P$ is a multiple of every $a_i$, for every $i \in [k] \setminus \{ 1, 2 \}$, there are at least $\frac{P}{a_i}$ schedulings of job $i$. The overall sum is 
$$\sum_{i = 3}^k \frac{P}{a_i} = -\frac{P}{a} - \frac{P}{b} + \sum_{i = 1}^k \frac{P}{a_i} = D(A) P - \frac{P}{a} - \frac{P}{b}$$
where we have used the definition of density. Since the total number of schedulings of all jobs is the number of schedulings for jobs 1 and 2 combined with that of the rest, using our generated formulas, it is
$$\frac{P}{a} + \frac{P}{b} + \frac{P}{ab^2} - 2 + D(A)P - \frac{P}{a} - \frac{P}{b} = \frac{P}{ab^2} - 2 + D(A)P$$

Recall that $D(A) \geq 1 - \frac{1}{ab^2} + \frac{1}{n}$ so
$$\frac{P}{ab^2} - 2 + D(A)P \geq \frac{P}{ab^2} - 2 + P \cdot \left( 1 - \frac{1}{ab^2} + \frac{1}{n} \right) = P + \frac{P}{n} - 2 > P$$
where we have used the fact that $P > 2n$ because $P$ is a multiple of $abn$, and both $a$ and $b$ are at least 2. We have proven that the total number of jobs scheduled is greater than $P$, but the number of available days is $P$, so by the Pigeonhole Principle, there must be at least one day for which two jobs are scheduled. This violates the pinwheel constraint, a contradiction. This completes the proof. 
\end{proof}

\subsection{\texorpdfstring{Proof of \cref{lemma:bgt2}.}{Proof of Lemma 5.2.}}\label{proof:bgt2}

\begin{proof}
We have verified programmatically that in any interval of length 24, there are at least two left-pushes (as defined in the lemma).

We can now prove the bound in essentially the same way as the lemma. Suppose, for the sake of contradiction, that there is some $A$ with $(a_1, a_2, a_3, a_4) = (3, 6, 6, 8)$ and $D(A) > \frac{95}{96}$ that is schedulable. Then, there is a valid schedule. By the Axiom of Archimedes, there exists $n \in \mathbb{N}$ such that $D(A) \geq \frac{95}{96} + \frac{1}{n}$.

Let $P = n \prod_{i = 1}^k a_i$, and consider the scheduling of jobs 1, 2, 3, and 4 between day 1 and day $P$. Partitioning these $P$ days into $\frac{P}{24}$ groups of $24$ contiguous days each, there are at least $2 \cdot \frac{P}{24}$ left-pushes. Our formula for the total number of schedulings of job 1, 2, 3, and 4 is
$$\frac{P}{3} + \frac{P}{6} + \frac{P}{6} + \frac{P}{8} + \left\lfloor \frac{w}{3} \right\rfloor + \left\lfloor \frac{x}{6} \right\rfloor + \left\lfloor \frac{y}{6} \right\rfloor + \left\lfloor \frac{z}{8} \right\rfloor \geq \frac{P}{3} + \frac{P}{6} + \frac{P}{6} + \frac{P}{8} + \frac{w}{3} - 1 + \frac{x}{6} - 1 + \frac{y}{6} - 1 + \frac{z}{8} - 1$$
$$= \frac{P}{3} + \frac{P}{6} + \frac{P}{6} + \frac{P}{8} + \frac{8w + 4x + 4y + 3z}{24} - 4$$

where $w$ is the number of left-pushes for job 1,  $x$ is the number of left-pushes for job 2, $y$ is the number of left-pushes of job 3 and $z$ is the number of left-pushes of job 4. Then, $w + x + y + z \geq \frac{P}{12}$, so
$$\frac{P}{3} + \frac{P}{6} + \frac{P}{6} + \frac{P}{8} + \frac{8w + 4x + 4y + 3z}{24} - 4 \geq \frac{P}{3} + \frac{P}{6} + \frac{P}{6} + \frac{P}{8} + \frac{3w + 3x + 3y + 3z}{24} - 4$$
$$= \frac{P}{3} + \frac{P}{6} + \frac{P}{6} + \frac{P}{8} + \frac{P}{96} - 4$$

Now, since $P$ is a multiple of every $a_i$, for every $i \in [k] \setminus \{ 1, 2, 3, 4 \}$, there are at least $\frac{P}{a_i}$ schedulings of job $i$. The overall sum is 
$$\sum_{i = 5}^k \frac{P}{a_i} = -\left( \frac{P}{3} + \frac{P}{6} + \frac{P}{6} + \frac{P}{8} \right) + \sum_{i = 1}^k \frac{P}{a_i} = D(A) P - \left( \frac{P}{3} + \frac{P}{6} + \frac{P}{6} + \frac{P}{8} \right)$$
where we have used the definition of density. Since the total number of schedulings of all jobs is the number of schedulings for jobs 1 and 2 combined with that of the rest, using our generated formulas, it is
$$\frac{P}{3} + \frac{P}{6} + \frac{P}{6} + \frac{P}{8} + \frac{P}{96} - 4 + D(A)P - \left( \frac{P}{3} + \frac{P}{6} + \frac{P}{6} + \frac{P}{8} \right) = \frac{P}{96} - 4 + D(A)P$$

Recall that $D(A) \geq \frac{95}{96} + \frac{1}{n}$ so
$$\frac{P}{96} - 4 + D(A)P \geq \frac{P}{96} - 4 + P \cdot \left( \frac{95}{96} + \frac{1}{n} \right) = P + \frac{P}{n} - 4 > P$$
where we have used the fact that $P > 4n$ because $P$ is a multiple of $n \cdot 3 \cdot 6 \cdot 6 \cdot 8$. We have proven that the total number of jobs scheduled is greater than $P$, but the number of available days is $P$, so by the Pigeonhole Principle, there must be at least one day for which two jobs are scheduled. This violates the pinwheel constraint, a contradiction. This completes the proof.
\end{proof}

\subsection{\texorpdfstring{Proof of \cref{lemma:bgt3}.}{Proof of Lemma 5.3.}}\label{proof:bgt3}

\begin{proof}
Each element of $B$ represents some set of elements in the original instance $A$. In addition, this represents a partition of the elements of $A$. To be more concrete, suppose for job $i$ with period $B_i$ there are corresponding jobs in $A$ with periods in $A$ of $A' = (A_{i1}, A_{i2}, \dots, A_{im})$. 

We first prove the following as a sublemma: For all $i \in [k']$, we can simultaneously pack jobs of periods $\lfloor \frac{9}{7} \cdot A_{i1} \rfloor, \lfloor \frac{9}{7} \cdot A_{i2} \rfloor, \dots, \lfloor \frac{9}{7} \cdot A_{im} \rfloor$ in the spaces allocated to a job of period $\lfloor \frac{9}{7} \cdot B_i \rfloor$. 

Consider a slight variant on the fold operation, $\pfold'$, in which, rather than running until every element is at most $\theta$, the operation stops when a single element remains. Under this definition, $\pfold'(A') = B_i$. In addition, the schedulability of $\pfold'(A')$ still implies the schedulability of $A'$. Now, let us round down each $A$ period down to the nearest period of the form $B_i \cdot 2^v$ for $v \in \mathbb{N}$. Formally, let
$$D_i = (f(A_{i1}), f(A_{i2}), \dots, f(A_{im})) \text{ where } f(x) = B_i \cdot 2^{\lfloor \log_2(x / B_i) \rfloor}$$

Let us define $d = \pfold'(D)$. Since the original sequence $D$ is composed entirely of numbers of the form $B_i \cdot 2^v$ and at each step of the fold operation we either introduce half of an existing number or decrease to the next lowest number, by induction, $d = B_i \cdot 2^{v_0}$ for some $v_0 \in \mathbb{Z}$. Now, note that $d \leq B_i$ because every element of $D$ is less than the corresponding element of $A$. In addition, for any pinwheel instance $A$ and any $\alpha \in \mathbb{R}_{+}$,
$$\pfold'(\alpha A) = \alpha \cdot \pfold'(A)$$

This is because the initial periods of $\alpha A$ are $\alpha$ times greater than the corresponding periods in $A$, and for every stage in the fold operation, this relationship is maintained (the same halving or decrease operation occurs up to scale factors). Let us define $r = \underset{j \in [m]}{\min} \frac{f(A_{ij})}{A_{ij}}$. Then, $r$ exists because there are a finite number of possibilities for $j$ and for every $j$, $\frac{f(A_{ij})}{A_{ij}} > \frac{1}{2}$ because $\lfloor \log_2 (x / B_i) \rfloor > \log_2 (x / B_i) - 1$. This means $r > \frac{1}{2}$. Therefore, 
$$d = \pfold'(D_i) \geq \pfold'(rA') = r \pfold'(A') > \frac{1}{2} \cdot B_i$$

Since $d \leq B_i$, $d > \frac{1}{2} \cdot B_i$ and $\frac{d}{B_i} = 2^{v_0}$ for some $v_0 \in \mathbb{Z}$, it must be that $d = B_i$. Now, let us define 
$$c = \underset{j \in [m]}{\min} \frac{\lfloor \frac{9}{7} \cdot D_{ij} \rfloor}{\frac{9}{7} \cdot D_{ij}}$$

We claim that 
$$c \geq \frac{\lfloor \frac{9}{7} \cdot B_i \rfloor}{\frac{9}{7} \cdot B_i}$$

This is true because every $D_{ij}$ is a multiple of $B_i$ and for every $n \in \mathbb{N}$ and $x \in \mathbb{R}_{+}$, $\lfloor nx \rfloor \geq n \lfloor x \rfloor$, so
$$\frac{\lfloor \frac{9}{7} \cdot D_{ij} \rfloor}{\frac{9}{7} \cdot D_{ij}} = \frac{\lfloor n \cdot \frac{9}{7} \cdot B_{ij} \rfloor}{\frac{9}{7} \cdot n \cdot B_{ij}} \geq \frac{n \lfloor \frac{9}{7} \cdot B_{ij} \rfloor}{\frac{9}{7} \cdot n \cdot B_{ij}} = \frac{\lfloor \frac{9}{7} \cdot B_{ij} \rfloor}{\frac{9}{7} \cdot B_{ij}}$$

This means that 
$$\pfold'\left( \left\lfloor \frac{9}{7} \cdot D_{i1} \right\rfloor, \left\lfloor \frac{9}{7} \cdot D_{i2} \right\rfloor, \dots, \left\lfloor \frac{9}{7} \cdot D_{im} \right\rfloor \right) \geq \left\lfloor \frac{9}{7} \cdot B_i \right\rfloor$$

Now, every element of $D$ is less than the corresponding element of $A_f$, so
$$\pfold'\left( \left\lfloor \frac{9}{7} \cdot A_{f1} \right\rfloor, \left\lfloor \frac{9}{7} \cdot A_{f2} \right\rfloor, \dots, \left\lfloor \frac{9}{7} \cdot A_{fm} \right\rfloor \right) \geq \left\lfloor \frac{9}{7} \cdot B_i \right\rfloor$$

This proves the sublemma. Now, for the overall proof, since the jobs in $A$ can be partitioned based on which job in $B$ they are mapped to, and space is allocated for all jobs in $C$ of period $\lfloor \frac{9}{7} \cdot B_i \rfloor$ for $i \in [k']$, appealing to the lemma, every job in $A'$ can be scheduled, meaning $A'$ is schedulable, as desired.
\end{proof}

\section*{Acknowledgments.}

I would like to gratefully acknowledge Prof. Robert Kleinberg for introducing me to the pinwheel problem and for his guidance. I would also like to acknowledge the helpful suggestions of three anonymous reviewers. 

\bibliographystyle{plain}
\bibliography{ref}

@inproceedings{kawamura,
author = {Kawamura, Akitoshi},
title = {Proof of the Density Threshold Conjecture for Pinwheel Scheduling},
year = {2024},
isbn = {9798400703836},
publisher = {Association for Computing Machinery},
address = {New York, NY, USA},
doi = {10.1145/3618260.3649757},
booktitle = {Proceedings of the 56th Annual ACM Symposium on Theory of Computing},
pages = {1816–1819},
numpages = {4},
location = {Vancouver, BC, Canada},
series = {STOC 2024}
}

@INPROCEEDINGS{05bound,
author={Holte, R. and Mok, A. and Rosier, L. and Tulchinsky, I. and Varvel, D.},
booktitle={[1989] Proceedings of the Twenty-Second Annual Hawaii International Conference on System Sciences. Volume II: Software Track}, 
title={The pinwheel: a real-time scheduling problem}, 
year={1989},
volume={2},
pages={693-702 vol.2},
doi={10.1109/HICSS.1989.48075}
}

@article{23bound,
author = {Chan, M. Y. and Chin, Francis},
date = {1993/05/01},
doi = {10.1007/BF01187034},
id = {Chan1993},
isbn = {1432-0541},
journal = {Algorithmica},
number = {5},
pages = {425--462},
title = {Schedulers for larger classes of pinwheel instances},
volume = {9},
year = {1993},
}

@ARTICLE{710bound,
author={Chan, M.Y. and Chin, F.Y.L.},
journal={IEEE Transactions on Computers}, 
title={General schedulers for the pinwheel problem based on double-integer reduction}, 
year={1992},
volume={41},
number={6},
pages={755-768},
doi={10.1109/12.144627}
}

@article{34bound,
author = {P. C. Fishburn and J. C. Lagarias},
date = {2002/09/01},
doi = {10.1007/s00453-002-0938-9},
id = {Fishburn2002},
isbn = {1432-0541},
journal = {Algorithmica},
number = {1},
pages = {14--38},
title = {Pinwheel Scheduling: Achievable Densities},
volume = {34},
year = {2002},
}

@article{kawamura2020,
author = {Akitoshi Kawamura and Makoto Soejima},
doi = {10.1016/j.tcs.2020.07.037},
issn = {0304-3975},
journal = {Theoretical Computer Science},
pages = {195-206},
title = {Simple strategies versus optimal schedules in multi-agent patrolling},
volume = {839},
year = {2020},
}

@misc{regular,
title={Quasi-Regular Sequences}, 
author={Joshua Frisch and Wade Hann-Caruthers and Pooya Vahidi Ferdowsi},
year={2019},
eprint={1909.13320},
archivePrefix={arXiv},
primaryClass={math.CO},
url={https://arxiv.org/abs/1909.13320}, 
}

@InProceedings{covering,
author="Kawamura, Akitoshi and Kobayashi, Yusuke and Kusano, Yosuke",
editor="Finocchi, Irene
and Georgiadis, Loukas",
title="Pinwheel Covering",
booktitle="Algorithms and Complexity",
year="2025",
publisher="Springer Nature Switzerland",
address="Cham",
pages="185--199",
isbn="978-3-031-92935-9",
doi={10.1007/s00453-002-0938-9}
}

@article{gasieniec,
title = {Perpetual maintenance of machines with different urgency requirements},
journal = {Journal of Computer and System Sciences},
volume = {139},
pages = {103476},
year = {2024},
issn = {0022-0000},
doi = {10.1016/j.jcss.2023.103476},
author = {Leszek Gąsieniec and Tomasz Jurdziński and Ralf Klasing and Christos Levcopoulos and Andrzej Lingas and Jie Min and Tomasz Radzik},
}

@inproceedings{bgtintro,
title={Bamboo garden trimming problem (perpetual maintenance of machines with different attendance urgency factors)},
author={G{\k{a}}sieniec, Leszek and Klasing, Ralf and Levcopoulos, Christos and Lingas, Andrzej and Min, Jie and Radzik, Tomasz},
booktitle={International Conference on Current Trends in Theory and Practice of Informatics},
pages={229--240},
year={2017},
organization={Springer},
doi={10.1007/978-3-319-51963-0_18}
}

@misc{bgt188,
title={An enhanced pinwheel algorithm for the bamboo garden trimming problem}, 
author={Federico Della Croce},
year={2020},
eprint={2003.12460},
archivePrefix={arXiv},
primaryClass={cs.DS},
url={https://arxiv.org/abs/2003.12460}, 
}

@inproceedings{bgt107,
author = {Felix H{\"{o}}hne and Rob van Stee},
editor = {Nicole Megow and Adam D. Smith},
title = {A 10/7-Approximation for Discrete Bamboo Garden Trimming and Continuous Trimming on Star Graphs},
booktitle = {Approximation, Randomization, and Combinatorial Optimization. Algorithms and Techniques, {APPROX/RANDOM} 2023, September 11-13, 2023, Atlanta, Georgia, {USA}},
series = {LIPIcs},
volume = {275},
pages = {16:1--16:19},
publisher = {Schloss Dagstuhl - Leibniz-Zentrum f{\"{u}}r Informatik},
year = {2023},
doi = {10.4230/LIPICS.APPROX/RANDOM.2023.16},
bibsource = {dblp computer science bibliography, https://dblp.org}
}

@article{bgt127,
title = {A 12/7-approximation algorithm for the discrete Bamboo Garden Trimming problem},
journal = {Operations Research Letters},
volume = {49},
number = {5},
pages = {645-649},
year = {2021},
issn = {0167-6377},
doi = {10.1016/j.orl.2021.07.001},
author = {Martijn {van Ee}},
}

@inbook{vanilla,
author = {William Kuszmaul},
title = {Achieving Optimal Backlog in the Vanilla Multi-Processor Cup Game},
booktitle = {Proceedings of the 2020 ACM-SIAM Symposium on Discrete Algorithms (SODA)},
pages = {1558-1577},
doi = {10.1137/1.9781611975994.96},
publisher = {Proceedings of the Fourteenth Annual ACM-SIAM Symposium on Discrete Algorithms, SODA '20},
year = {2020},
}

@InProceedings{variable,
author = {Kuszmaul, William and Westover, Alek},
title =	{{The Variable-Processor Cup Game}},
booktitle =	{12th Innovations in Theoretical Computer Science Conference (ITCS 2021)},
pages =	{16:1--16:20},
series =	{Leibniz International Proceedings in Informatics (LIPIcs)},
ISBN =	{978-3-95977-177-1},
ISSN =	{1868-8969},
year =	{2021},
volume =	{185},
editor =	{Lee, James R.},
publisher =	{Schloss Dagstuhl -- Leibniz-Zentrum f{\"u}r Informatik},
address =	{Dagstuhl, Germany},
URN =		{urn:nbn:de:0030-drops-135559},
doi =		{10.4230/LIPIcs.ITCS.2021.16}
}

@article{layland,
author = {Liu, C. L. and Layland, James W.},
title = {Scheduling Algorithms for Multiprogramming in a Hard-Real-Time Environment},
year = {1973},
issue_date = {Jan. 1973},
publisher = {Association for Computing Machinery},
address = {New York, NY, USA},
volume = {20},
number = {1},
issn = {0004-5411},
doi = {10.1145/321738.321743},
journal = {J. ACM},
month = jan,
pages = {46–61},
numpages = {16}
}

@inbook{towards,
author = {Leszek G{\k a}sieniec and Benjamin Smith and Sebastian Wild},
booktitle = {2022 Proceedings of the Symposium on Algorithm Engineering and Experiments (ALENEX)},
doi = {10.1137/1.9781611977042.8},
eprint = {https://epubs.siam.org/doi/pdf/10.1137/1.9781611977042.8},
pages = {91-103},
title = {Towards the 5/6-Density Conjecture of Pinwheel Scheduling},
year = {2022},
publisher = {siam}
}

@inproceedings{multip,
author = {Bender, Michael A. and Farach-Colton, Mart\'{\i}n and Kuszmaul, William},
title = {Achieving optimal backlog in multi-processor cup games},
year = {2019},
isbn = {9781450367059},
publisher = {Association for Computing Machinery},
address = {New York, NY, USA},
doi = {10.1145/3313276.3316342},
booktitle = {Proceedings of the 51st Annual ACM SIGACT Symposium on Theory of Computing},
pages = {1148–1157},
numpages = {10},
location = {Phoenix, AZ, USA},
series = {STOC 2019}
}

@article{priority,
AUTHOR = {D’Emidio, Mattia and Di Stefano, Gabriele and Navarra, Alfredo},
TITLE = {Bamboo Garden Trimming Problem: Priority Schedulings},
JOURNAL = {Algorithms},
VOLUME = {12},
YEAR = {2019},
NUMBER = {4},
ARTICLE-NUMBER = {74},
ISSN = {1999-4893},
DOI = {10.3390/a12040074}
}

@inproceedings{revisited,
author = {Kuszmaul, John},
title = {Bamboo Trimming Revisited: Simple Algorithms Can Do Well Too},
year = {2022},
isbn = {9781450391467},
publisher = {Association for Computing Machinery},
address = {New York, NY, USA},
doi = {10.1145/3490148.3538580},
booktitle = {Proceedings of the 34th ACM Symposium on Parallelism in Algorithms and Architectures},
pages = {411–417},
numpages = {7},
location = {Philadelphia, PA, USA},
series = {SPAA '22}
}

@article{applications,
title = {Fifty years of research in scheduling — Theory and applications},
journal = {European Journal of Operational Research},
volume = {327},
number = {2},
pages = {367-393},
year = {2025},
issn = {0377-2217},
doi = {10.1016/j.ejor.2025.01.034},
author = {Alessandro Agnetis and Jean-Charles Billaut and Michael Pinedo and Dvir Shabtay},
}

@misc{kawamura2025computerassistedproofoptimaldensity,
      title={A Computer-Assisted Proof of the Optimal Density Bound for Pinwheel Covering}, 
      author={Akitoshi Kawamura and Yusuke Kobayashi},
      year={2025},
      eprint={2510.06533},
      archivePrefix={arXiv},
      primaryClass={cs.DM},
      url={https://arxiv.org/abs/2510.06533}, 
}

\end{document}